\documentclass{article}
\usepackage{rotating}
\usepackage{amsmath}
\usepackage{amsfonts}
\usepackage{amssymb}
\usepackage{tikz}
\newtheorem{Theorem}{Theorem}[section]

\newtheorem{thm}[Theorem]{Theorem}
\newtheorem{dfn}[Theorem]{Definition}

\newtheorem{lem}[Theorem]{Lemma}
\newtheorem{eg}[Theorem]{Example}

\newenvironment{proof}[1][Proof]{\noindent\textbf{#1.} }{\ \rule{0.5em}{0.5em}}
\newcommand{\bpartial}{\mathop{\partial\kern -4pt\raisebox{.8pt}{$|$}}}
\newcommand{\bra}{\mathopen{[\kern-1.6pt[}}
\newcommand{\ket}{\mathclose{]\kern-1.5pt]}}
\newcommand{\bbra}{\mathopen{[\kern-2.2pt[\kern-2.3pt[}}
\newcommand{\bket}{\mathclose{]\kern-2.1pt]\kern-2.3pt]}}

\makeindex
\makeatletter
\newcommand*{\rom}[1]{\expandafter\@slowromancap\romannumeral #1@}
\makeatother

\begin{document}

\title {\large{\bf
 Integrable bi-Hamiltonian systems by
 Jacobi structure
on  real three-dimensional
Lie groups
 }}
\vspace{3mm}
\author { \small{ \bf H. Amirzadeh-Fard$^{1, 3}$ }\hspace{-2mm}{ \footnote{ e-mail: h.amirzadehfard@azaruniv.ac.ir}},
 \small{ \bf  Gh. Haghighatdoost$^1$}\hspace{-1mm}{ \footnote{ e-mail: gorbanali@azaruniv.ac.ir}},
 \small{ \bf A.
Rezaei-Aghdam$^2$ }\hspace{-1mm}{\footnote{ e-mail: rezaei-a@azaruniv.ac.ir }}\\
{\small{$^{1}$\em
 Department
of Mathematics, Azarbaijan Shahid Madani University, 53714-161, Tabriz, Iran}}\\
{\small{$^{2}$\em Department of Physics, Azarbaijan Shahid Madani
University, 53714-161, Tabriz, Iran}}\\
{\small{$^{3}$\em
Department of Education,  Beiranvand  Region,
 99904038, Lorestan, Iran}}\\ }
 \maketitle
\begin{abstract}
By Poissonization of  Jacobi structures on  real
three-dimensional Lie groups $\mathbf{G}  $ and using the realizations of
their Lie algebras, we obtain integrable bi-Hamiltonian systems  on
$\mathbf{G}   \otimes \mathbb{R}.$
\end{abstract}

\smallskip

{\bf keywords:}{  Poissonization, Jacobi manifold, Completely integrable
Hamiltonian system.}

%%%%%%%%%%%%%%%%%%%%%%%%%%%%%%%%%%%%%%%%%%%%%%%%%%%%%%%%%%%%%%%%%%%%%%%%%%%%%%%%%%
\section {\large {\bf Introduction}}
Let $(M^{2n},\omega)$ be a  symplectic manifold, and  $\dot{x}=X_H(x)$ be a
 Hamiltonian
 system on it,
where $H$ is a smooth function on $M^{2n}$ which is called the
Hamiltonian and   vector field $X_H(x)=\omega^{-1}(dH(x))$
is the  corresponding Hamiltonian  vector field.
A completely Liouville-integrable
Hamiltonian system
is a Hamiltonian system
 with $n$  functionally independent first integrals  in involution.
  In other words, two smooth functions f and g on
 $(M^{2n},\omega)$
 are in involution if their Poisson bracket equals
zero \textcolor{red}{ [\ref{Bolsinov}]}.
 A symmetry  of  a Hamiltonian system $(M^{2n}, \omega, H)$
is a transformation $S:M^{2n}\longrightarrow M^{2n}$
such that
$S^*\omega=\omega$
and
  $S^*H=H $
where $S^*$
is the pullback of the symplectic form $\omega$ or $H$ by $S.$
The set of all symmetries forms
a group, which is called the symmetry group and can be a Lie group  \textcolor{red}{[\ref{symm}]}.
 One can construct a dynamical system
 for which
 the Lie group
 plays the role of the symmetry Lie group
 and
 the symplectic manifold plays  the role of the phase space
 \textcolor{red}{[\ref{Abedi-Fardad}]}.

 Jacobi manifolds which are a generalization of  Poisson   manifolds
 have various applications in  physics and  classical mechanics.
 Poisson manifolds  play an important role as
 phase spaces of classical mechanics.
  In \textcolor{red}{[\ref{hassan2}]}, we  have  classified all Jacobi structures on real three-dimensional Lie groups.
In this work, using  the Poissonization  of the Jacobi structure $( \Lambda, E)$ on $M$
\textcolor{red}{[\ref{ref4}]},
we convert  the Jacobi structure $( \Lambda, E)$ on $M$
  into
 the Poisson structure
% $P=e^{-s}(\Lambda+\partial_s\wedge E)$
on $M\times \mathbb{R}$
%where $s$ is the coordinate on $ \mathbb{R}\, .$
and we consider those
  Poisson structures that are non-degenerate, and so define  symplectic  structures.
Then,
applying Darboux's theorem \textcolor{red}{ [\ref{Bolsinov}]} and
using realizations \textcolor{red}{ [\ref{realiz}]} of
 real  three-dimensional Lie algebras
 $
 \mathfrak{g}$
(related
  Lie group $\mathbf{G}  $),
  we construct  integrable  Hamiltonian systems
   for which the
 Lie group $\mathbf{G}  $
 plays the role of  the symmetry Lie group
 and
 the symplectic manifold $\mathbf{G}  \otimes \mathbb{R}$ plays the role of the phase space.

 In mathematical physics and mechanics,   many integrable dynamical systems   admit the
 bi-Hamiltonian  structure,
that is Hamiltonian with respect to two compatible Poisson  structures $P_1$ and $P_2$
\textcolor{red}{ [\ref{Magri}]}.
%\textcolor{red}{ [\ref{Magri}]}.
Here, we  will calculate
  all equivalence-classes of
Jacobi structures on real three-dimensional Lie groups
$\mathbf{G}$
 for which after  Poissonization we have  non-degenerate Poisson brackets on
  $M=\mathbf{G} \otimes \mathbb{R}$
  and then we will
   study
 the existence of a bi-Hamiltonian structure for a completely
integrable Hamiltonian system.

The outline of the paper is as follows: In Sec. 2,
 we  briefly recall
 the Jacobi structures on real low-dimensional Lie
groups
 and also
 the construction of
 the Liouville-integrable
Hamiltonian system.
In Sec. 3,
 we find integrable Hamiltonian systems such
 that their phase spaces
 are obtained by
 using
  Poissonization of
  the Jacobi structures on some real three-dimensional Lie groups.
 In Sec. 4,
   we study
 the existence of the bi-Hamiltonian structure for a completely
integrable Hamiltonian system obtained    in Sec. 3.
%%%%%%%%%%%%%%%%%%%%%%%%
\newpage
\section{A  review of the necessary constructions
}
For self-containing  of the paper,
 we review the essential results about Jacobi structures on real low-dimensional Lie
groups
\textcolor{red}{[\ref{hassan2}]}
and Integrable Hamiltonian systems \textcolor{red}{[\ref{Bolsinov}, \ref{Abedi-Fardad}]}.

\subsection{ Jacobi structures on real low-dimensional Lie
groups}
The study of  the Jacobi manifolds was introduced   by Lichnerowicz and Kirillov \textcolor{red}{[\ref{ref4}, \ref{ref3}]}.\\
A Jacobi manifold  $(M, \mathbf{\Lambda },\mathbf{ E})$  is a
manifold $M$
admitting a
 bivector field
 $\mathbf{\Lambda }$
 and
 a Reeb vector
field
$\mathbf{E}$
 such that
$
[\mathbf{\Lambda }, \mathbf{\Lambda }]=2E\wedge \mathbf{\Lambda },\,
L_{E}\mathbf{\Lambda }=[E, \mathbf{\Lambda }]=0,
$
where [.,.] stands for
the  Schouten-Nijenhuis  bracket
\footnote{
\text{
For general p-vector fields} $
X_1\wedge...\wedge X_p$
\text{and
q-vector fields}
$Y_1\wedge...\wedge Y_q,$
 \text{the Schouten-Nijenhius bracket
 is given by}

 $[X_1\wedge...\wedge X_p, Y_1\wedge...\wedge Y_q]=
 (-1)^{p+1}
\sum\limits_{{\rm{i = 1}}}^{\rm{p}} \sum\limits_{{\rm{j = 1}}}^{\rm{q}}(-1)^{i+j}[X_i,Y_j] \wedge X_1\wedge...\wedge\hat{X_i}\wedge...\wedge X_p\wedge Y_1\wedge...\wedge\hat{Y_j}\wedge...\wedge Y_q,$
\text{for all} $X_i, Y_i \in \frak{X}(M),$
\text{where} $[X_i,Y_j]$\text{ denotes the Lie bracket of the two vector fields}
$X_i,\,Y_j$ \text{on M}
 \textcolor{red}{[\ref{Vaisman}]}.}
 \textcolor{red}{[\ref{ref4}]}.
If $(M, \mathbf{\Lambda },\mathbf{ E})$ is a  Jacobi manifold, then
the space $ ( C^{\infty}(M, \mathbb{R}), \lbrace .,.\rbrace_{ \Lambda,E})$ becomes
a local Lie algebra in the sense of Kirillov \textcolor{red}{[\ref{ref3}]} with the
 following
Jacobi bracket
\begin{equation}
\lbrace f,g\rbrace_{ \mathbf{\Lambda },\mathbf{E}}= \mathbf{\Lambda }(df,dg)+f \mathbf{E} g-g\mathbf{ E} f,\quad \forall f,g \in C^{\infty}(M).
\end{equation}
This Lie bracket is a Poisson bracket
if and only if the vector field
$E$ identically vanishes.

As shown by Lichnerowicz \textcolor{red}{ [\ref{ref4}]},
to any Jacobi manifold
  $(M, \mathbf{\Lambda },\mathbf{ E})$
  one can associate a Poisson  manifold
$(M\otimes \mathbb{R}, P)$
 with the  Poisson bivector $P$ as:
\begin{equation}\label{aghajon}
P=e^{-s}(\mathbf{\Lambda}+\partial_s\wedge\mathbf{ E})
\end{equation}
where $s$ is the coordinate on $ \mathbb{R}$. The  Poisson manifold $(M\otimes \mathbb{R}, P)$ is said to be the Poissonization of the Jacobi manifold  $(M, \mathbf{\Lambda}, \mathbf{E})$.

 Let $x^{\mu}(\mu=1, . . . ,dimM)$
 be the local coordinates chart of a  Jacobi manifold  $M$,
then the tensor field $ \mathbf{\Lambda }$, the vector field $ \mathbf{E}$ and the Jacobi bracket on $M$ can be written
as follows:
\begin{equation}\label{}
\mathbf{\Lambda } = \frac{1}{2}{\mathbf{\Lambda} ^{\mu\nu}}\partial_ {\mu}\wedge \partial_ {\nu},
\end{equation}
\begin{equation}\label{}
\mathbf{ E }= \mathbf{E}^{\mu} \partial _{\mu},
\end{equation}
\begin{equation}\label{s4}
\{ f,g\}_{\mathbf{ \Lambda, E} } = {\mathbf{\Lambda} ^{\mu\nu}}\partial_ {\mu}f\partial_ {\nu}g + f{\mathbf{ E}}^{\mu}\partial_ {\mu} g - g\mathbf{ E}^{\mu}\partial _{\mu}f,\quad \forall f,g \in C^{\infty}(M).
\end{equation}
Furthermore,  by substituting   the Jacobi bracket (\ref{s4}) in the  Jacobi identity, one can  obtain the following relations
\begin{equation}\label{s5}
{\mathbf{\Lambda} ^{\nu \rho }}{\partial _\rho }{\mathbf{\Lambda} ^{\lambda \mu } + \mathbf{\Lambda} ^{\mu \rho }{\partial _\rho }{\mathbf{\Lambda} ^{\nu \mathbf{\lambda} }} + {\mathbf{\Lambda }^{\mathbf{\lambda} \rho }}{\partial _\rho }{\mathbf{\Lambda }^{\mu \nu }} + \mathbf{E}^\mathbf{\lambda} }{\mathbf{\Lambda }^{\mu \nu }} + \mathbf{E}^{\mu }{\mathbf{\Lambda} ^{\nu \mathbf{\lambda }}} + \mathbf{ E}^{\nu }{\mathbf{\Lambda} ^{\mathbf{\lambda} \mu }} = 0,
\end{equation}
\begin{equation}\label{s6}
\mathbf{E}^{\rho }{\partial _\rho }{\mathbf{\Lambda} ^{\mu \nu }} - {\mathbf{\Lambda} ^{\rho \vartheta }}{\partial _\rho }\mathbf{ E}^{\mu } + {\mathbf{\Lambda }^{\rho \mu }}{\partial _\rho }\mathbf{ E}^{\nu } = 0.
\end{equation}
The Eqs. (\ref{s5}) and (\ref{s6} ) are called the Jacobi equations.
The general solution
for the Jacobi equations yields the general form of the Jacobi structures on a manifold $M$
\textcolor{red}{ [\ref{hass}]}.
We have obtained  the Jacobi structures on  real three-dimensional Lie  groups
as a smooth manifold in \textcolor{red}{[\ref{hassan2}]}.
To find these Jacobi structures, one must determine the vielbein
$e_{a}^{\;\;\mu}$ for the Lie
groups, and for this, one must  find the left-invariant one-form on the Lie group:
 \begin{equation}
 g^{-1}dg= e_{\;\;\mu}^{a}X_{a}dx^{\mu},\qquad \forall g\in \mathbf{G}
 \end{equation}
where $\lbrace X_{a} \rbrace$ are generators of the Lie group. All left-invariant one-forms
on  real three-dimensional Lie groups
 were previously obtained in  \textcolor{red}{ [\ref{Hemmati}]}.
Therefore, one can compute
the inverse of the vielbein
$e_{\;\;\mu}^{a}$
 (i.e.,
$e_{a}^{\;\;\mu}$ with
$e_{\;\;\mu}^{a }   e_{a}^{\;\;\nu} =\delta_{\mu} ^{\;\;\nu},  e_{\;\;\mu}^{a } e_{b}^{\;\;\mu} =\delta_{\;\;b} ^{a})
$
 for Lie groups
using  left-invariant one-forms.
Note that
the elements of
% real two-dimensional Lie  group $\mathbf{G}  $
% are given by
%$g=e^{xX_{1}}e^{yX_{2}}$
%for all
%$g \in\mathbf{G}.$
%Additionally, for
 real three-dimensional
 Lie  group $\mathbf{G}  $
 are given by
$g=e^{xX_{1}}e^{yX_{2}}e^{zX_{3}}
$
for all
$g \in\mathbf{G},$
where
  $( x, y, z )$ is
the local coordinate system on  the Lie group  $\mathbf{G}  $.
 The Jacobi structure
 $(\mathbf{G, \Lambda, E)}$
 on  the Lie group  $\mathbf{G}  $
  is
written in terms of the
 non-coordinate basis
 \footnote{
 The bases $\lbrace\hat{e}_a\rbrace$ and $\lbrace\hat{\theta}^{a}\rbrace$ are called the non-coordinate bases \textcolor{red}{ [\ref{naka}]}.
 }
  as
\begin{equation}\label{basis}
{\mathbf{\Lambda }^{\mu \nu }}= e_{a}^{\;\;\mu} e_{b}^{\;\;\nu} {\Lambda ^{ab }},
\end{equation}
\begin{equation}\label{basis2}
\mathbf{ E}^{\mu}= e_{a}^{\;\;\mu}E^{a},
\end{equation}
where
 $\Lambda^{ab} $  and $E^a$  are
   Jacobi   structures
 on Lie algebra
 $
 \mathfrak{g}$
  and
 we have assumed that these are independent
of the coordinate of  the Lie group,
 and
the indices $ \mu, \nu, \cdots $ and  $a, b, \cdots$ are respectively related to the Lie group
coordinates and the Lie algebra basis.

Taking into account
that  $\hat{e}_a= e_{a}^{\;\;\mu} \; \partial _\mu$,  then we have
\textcolor{red}{[\ref{naka}]}
\begin{equation}
[\hat{e}_a, \hat{e}_b]=\mathbf{f}_{ab}^{\;\;c}\hat{e}_c,
\end{equation}
where
 $\mathbf{f}_{ab}^{\;\;c}$ (i.e.,
 the structure constants of the Lie algebra $\mathfrak{g}$)
are related to the vielbein
$ e_{a}^{\;\;\mu}$
 by the
Maurer-Cartan
relation \textcolor{red}{[\ref{naka}]}:
\begin{equation}\label{Maurer-Cartan}
{\mathbf{f}}_{ab}^{\;\;c}=e^c_{\;\;\nu}( e_{a}^{\;\;\mu}\partial _\mu e_{b}^{\;\;\nu}-e_{b}^{\;\;\mu}\partial _\mu e_{a}^{\;\;\nu}).
\end{equation}
Inserting Eqs.
(\ref{basis}) and (\ref{basis2}) into  Eqs. (\ref{s5}) and (\ref{s6}) and using the Maurer-Cartan equation,
one can obtain \textcolor{red}{[\ref{hassan2}]}:
\begin{equation}\label{lie 17}
{{\mathbf{f}}_{bc}}^f{\Lambda ^{hb}}{\Lambda ^{ce}} + {{\mathbf{f}}_{bd}}^e{\Lambda ^{hb}}{\Lambda ^{fd}} + {{\mathbf{f}}_{ba}}^h{\Lambda ^{eb}}{\Lambda ^{af}} + {{\rm E}^f}{\Lambda ^{eh}} + {{\rm E}^e}{\Lambda ^{hf}} + {{\rm E}^h}{\Lambda ^{fe}} = 0,
\end{equation}
\begin{equation}\label{lie 16}
{{\mathbf{f}}_{ac}}^d{{\rm E}^a}{\Lambda ^{ce}} + {{\mathbf{f}}_{ab}}^e{{\rm E}^a}{\Lambda ^{db}} = 0.
\end{equation}
It is quite difficult to get results working with the tensor form of Eqs. (\ref{lie 17}) and (\ref{lie 16}); thus we  propose
using the adjoint representations of Lie algebras

\begin{equation}\label{rep}
{{\mathbf{f}}_{ab}}^c =  - {({\chi _a})_b}^c,\qquad
{{\mathbf{f}}_{ab}}^c =  - {({\cal Y}^c)_{ab}},\qquad
%{{\mathbf{f}}_{ab}}^c =  - {{\mathbf{f}}_{ba}}^c.
\end{equation}
then the Eqs.   (\ref{lie 17}) and (\ref{lie 16}) in the matrix form can be rewritten  respectively as follows \textcolor{red}{[\ref{hassan2}]}:
\begin{equation}\label{Lie al2}
 -\Big ({\Lambda ^{ce}}({\chi ^t}_c\Lambda ) + \Lambda {\cal Y}^e\Lambda  + (\Lambda \chi _b){\Lambda ^{be}} + {{\rm E}^e}\Lambda \Big) ^{fh}+ {{\rm E}^f}{\Lambda ^{eh}} + {\Lambda ^{fe}}{{\rm E}^h} = 0,
\end{equation}
\begin{equation}\label{Lie al1}
(\Lambda \chi _a - {(\Lambda {\chi _a})^t}){{\rm E}^a} = 0.
\end{equation}
 The general solution
of   Eqs. (\ref{Lie al2}) and (\ref{Lie al1}) yields the general form of the Jacobi structures.
 In order to find general solutions of these equations, one can use the $Maple$ program.
 %For more details, see [\ref{hassan2}].
Applying  Eqs. $(\ref{basis})$ and $(\ref{basis2})$, one can  obtain the
Jacobi structures
$\mathbf{\Lambda}$ and $\mathbf{E}$
on the Lie group.
In \textcolor{red}{[\ref{hassan2}]}
we have obtained all
Jacobi structures on three-dimensional Lie algebras and
their
Lie groups.
Here,
 we will consider
 those structures such that
 after Poissonization of them
 (see relation (\ref{aghajon}) )
 the resulting
 Poisson structures
 are  nondegenerate. The results are given
in  Table 1. We will see that
only  the Lie groups
$\mathbf{\rom{2}} \otimes \mathbb{R}, ~ \mathbf{ \rom{3} }\otimes \mathbb{R}, ~ \mathbf{\rom{4} }\otimes \mathbb{R},~
\mathbf{\rom{6}_0} \otimes \mathbb{R},~ \mathbf{\rom{7}_0 }\otimes \mathbb{R}
$
 have nondegenerate Poisson structures  (see Table  1).
\subsection{Liouville-integrable
Hamiltonian systems
with symmetry Lie groups}

Let $(M^{2n},\omega_{ij})$ be a  symplectic manifold,  and let $(x_1,\cdots, x_{2n})$  be the local coordinates system
on $M^{2n}$
as a phase space.
The relationship between the Poisson bracket  on the space of smooth functions on $M^{2n}$ and the   symplectic form
$\omega_{ij}$
 is given by
$\lbrace f, g\rbrace =P^{ij}\frac{\partial f}{\partial x_i}\frac{\partial g}{\partial x_j}$
where
$P^{ij}$   is the inverse of 2-form  $\omega_{ij}.$
\begin{thm}\label{}
${\it (G. Darboux)}$ \textcolor{red}{[\ref{Bolsinov}]} For any point of a symplectic manifold $(M^{2n},\omega),$
there exists an open neighborhood possessing
canonical  coordinate $(q_1,\cdots, q_n, p_1,\cdots, p_n ) $
in which the symplectic structure $\omega$ admits  the canonical form $\omega=\sum\limits_{{\rm{i = 1}}}^{\rm{n}}dq_i\wedge dp_i$. In other words,
the canonicity condition for the symplectic structure can
 be rewritten in terms of the Poisson bracket
 as follows:
\begin{equation*}
\lbrace p_i,p_j   \rbrace=0,\qquad
\lbrace q_i,q_j   \rbrace=0,\qquad
\lbrace q_j ,p_i  \rbrace=\delta_{ij}\qquad i,j=1,\cdots, n.
\end{equation*}
\end{thm}

 Given a dynamical system
 for which
 Lie group
 $\mathbf{G}  $
 plays the role of symmetry group
 and
symplectic manifold
 $M^{2n}$
plays the of the phase space,
one can construct
independent  dynamical  functions
$S_l= S_l (q_1,\cdots, q_n, p_1,\cdots, p_n )$  on  the phase space $M^{2n}$
  satisfying
   \textcolor{red}{[\ref{Abedi-Fardad}]}
\begin{equation}\label{canon}
\lbrace S_i, S_j  \rbrace=\sum\limits_{{\rm{n = 1}}}^{\rm{2}}
\Big(
\frac{\partial S_i}{\partial q_n}\frac{\partial S_j}{\partial p_n}
-\frac{\partial S_i}{\partial p_n}\frac{\partial S_j}{\partial q_n}
\Big)=f_{ij}^k S_k,
  \end{equation}
where $f_{ij}^k$ stand for the  structure constants of the Lie algebra $\mathfrak{g}$
associated with the symmetry   Lie  group  $\mathbf{G}$.\\
Using  the relation(\ref{canon}),
one can
find
the number of integrals  $S_l$  which commute with respect to the Poisson bracket related to the symplectic form. In other words,
$\lbrace S_i,S_j\rbrace=0,  \,i,j=1,\cdots,n;$
   such that one of the functions
 $ S_l  $ can be considered
as a Hamiltonian of the integrable system
\textcolor{red}{[\ref{Abedi-Fardad}]}.

In the following section, we will use  Jacobi structures
associated with real three-dimensional Lie groups $\mathbf{G}$  to construct the Poisson
structure
 on
 $\mathbf{G} \otimes \mathbb{R}$ (Poissonization) and then will obtain related
  integrable Hamiltonian systems. We will perform those using of
 the differential realization of real three-dimensional Lie groups  \textcolor{red}{ [\ref{realiz}]}.

\section{
 Integrable Hamiltonian systems
  by Jacobi structures on  real
  three-dimensional Lie groups
}\label{section3}
In this section, we shall consider the different representations of one
 equivalence-class  of
Jacobi structures on real three-dimensional Lie groups
$\mathbf{G}$
 for which after  Poissonization
  we have a non-degenerate Poisson bracket on  $\mathbf{G} \otimes \mathbb{R}$
  for each representation of equivalence-classes
  \footnote{
Note that previously in \textcolor{red}{ [\ref{hassan2}]}
 we have obtained all equivalence classes of Jacobi structure on a three-dimensional Lie group. Here we use only one of those equivalence classes, because on other classes after Poissonization the obtained Poisson brackets are degenerate or singular \textcolor{red}{ [\ref{hassanone}]}.}.
 To simplify the presentation, we will discuss integrable Hamiltonian systems only for
  one
  representation
 of
 equivalence classes.
\begin{eg}
  Lie group
$\mathbf{\rom{2}}$

Considering
  the  Lie group
$\mathbf{\rom{2}}$  related to  the Lie algebra
${\rom{2}}$ with non-zero commutators
$[X_2, X_3]=X_1,$
  it admits the Jacobi
structure as follows (see Table  1):
\begin{equation*}\label{moh}
\mathbf{\Lambda_{1}}=-z\partial _x\wedge \partial _z+\partial _y\wedge \partial _z,
\quad
\mathbf{E_{1}}=-\partial _x
\end{equation*}
where
  $( x, y, z )$ is
the local coordinate system on  the Lie group $\mathbf{\rom{2}}$.
Applying the Poissonization (\ref{aghajon})  of the Jacobi manifold  $(\mathbf{\rom{2}}, \mathbf{\Lambda_{1}},\mathbf{ E_{1}}),$
  it leads to the Poisson manifold
  $(\mathbf{\rom{2}} \otimes \mathbb{R}, P_1)$
    with
    the   local  coordinate system  $ (x, y, z, s) $ and
     the non-degenerate Poisson structure
\begin{equation*}
P_1:\qquad
\lbrace x,z  \rbrace=-z{e}^{-s} ,\qquad
\lbrace  x,s \rbrace={e}^{-s},\qquad
\lbrace  y,z \rbrace={e}^{-s},
\end{equation*}
 Now one can find the following Darboux
coordinates:
$$
q_1=x,\qquad
q_2=y
,\qquad
p_1={e}^{s}  ,\qquad
p_2=z{e}^{s},
$$
such that they satisfy in the following canonical  Poisson brackets:
\begin{equation*}
\lbrace q_1,p_1   \rbrace=1,\qquad
\lbrace q_2,p_2  \rbrace=1.
\end{equation*}
We now consider that the Lie algebra $\rom{2}$ is realized by means of  smooth transformations  on the
phase space $\mathbb{  R}^4$  with the canonical  coordinate $(q_1, q_2, p_1, p_2 ) $
\begin{equation*}
 S_i=X_i(q_1, q_2, p_1,p_2 ),
\end{equation*}
where
 the  differential operator of
$ p_1=-\frac{\partial}{\partial q_1}$ and $ p_2=-\frac{\partial}{\partial q_2}$
(quantum mechanical realization)
are
 the conjugate momentums to
$ q_1$
 and
 $ q_2$, respectively.
Using the results of    \textcolor{red}{ [\ref{realiz}]} ( see Table 2)
 one can get the $S_i$
% (see relation (\ref{canon}))
 as follows:
\begin{equation*}
S_1=-p_1=-{e}^{s},\qquad
S_2=-p_2=-z {e}^{s} ,\qquad
S_3=- q_2 p_1= -y{e}^{s}.
\end{equation*}
In this  way, now
applying relation (\ref{canon}),
one can show that
 they satisfy the following Poisson brackets
\begin{equation*}
\lbrace S_2, S_3 \rbrace=S_1,
\end{equation*}
i.e. we have the integrable Hamiltonian  system with symmetry Lie group
$\mathbf{\rom{2}} $  such that one can consider
its Hamiltonian
as:
\begin{equation*}
H=S_3=-y{e}^{s}
\end{equation*}
and the invariants of the system are $(H, S_1).$
\end{eg}
\begin{eg}
  Lie group
$\mathbf{\rom{3}}$

Now consider
  the  Lie group
$\mathbf{\rom{3}}$  related to the  Lie algebra
${\rom{3}}$ with non-zero commutators
 $$[X_1, X_2]=-(X_2+X_3), [X_1, X_3]=-(X_2+X_3),$$
 it admits the Jacobi
structure as follows (see Table 1 ):
\begin{equation*}
\mathbf{\Lambda_1}=
\partial _x\wedge \partial _z
+ \left( y+z \right) \partial _y\wedge \partial _z
\qquad
\mathbf{E_1}=
\partial _y-\partial _z
 \end{equation*}
where
  $( x, y, z )$ is
the local coordinate system on  the Lie group $\mathbf{\rom{3}}$.
Applying the Poissonization (\ref{aghajon})  of the Jacobi manifold  $(\mathbf{\rom{3}},\mathbf{ \Lambda_{1}}, \mathbf{E_{1}}),$
  it leads to the Poisson manifold
  $(\mathbf{\rom{3}} \otimes \mathbb{R}, P_1)$
    with
    the   local  coordinate system  $ (x, y, z, s) $ and
     the non-degenerate Poisson structure
\begin{equation*}
P_1:\qquad
\lbrace x,z  \rbrace=e^{-s},\qquad
\lbrace  y,z \rbrace={e}^{-s}\left( y+z \right)  ,\qquad
\lbrace  y,s \rbrace=-{e}^{-s},\quad
\lbrace  z,s \rbrace={e}^{-s}.
\end{equation*}
Now one can find the following Darboux
coordinates:
\begin{equation*}
q_1=-{e}^{s}y,\qquad
q_2=x,\qquad
p_1=s,\qquad
p_2=  z {e}^{s}
\end{equation*}
such that they satisfy in the following canonical  Poisson brackets:
\begin{equation*}
\lbrace q_1,p_1   \rbrace=1,\qquad
\lbrace q_2,p_2  \rbrace=1.
\end{equation*}
We now consider that the Lie algebra $\rom{3}$ is realized by means of  smooth transformations  on the
phase space $\mathbb{  R}^4$  with the canonical  coordinate $(q_1, q_2, p_1, p_2 ) $
\begin{equation*}
 S_i=X_i(q_1, q_2, p_1,p_2 ),
\end{equation*}
where
 the  differential operator of
$ p_1=-\frac{\partial}{\partial q_1}$ and $ p_2=-\frac{\partial}{\partial q_2}$
(quantum mechanical realization)
are
 the conjugate momentums
  to
$ q_1$
 and
 $ q_2$, respectively.
Using the results of
 \textcolor{red}{ [\ref{realiz}]} ( see Table 2)
 one can get the $S_i$
% (see relation (\ref{canon}))
 as follows:
\begin{equation*}
 S_1=-(q_1+q_2)(p_1+p_2)=
 ({e}^{s}y-x)  \left( s+ z {e}^{s}\right),  \qquad
S_2=-p_1=
-s,\qquad
S_3=-p_2=-  z  {e}^{s}.
\end{equation*}
In this  way, now
applying relation (\ref{canon}),
one can show that
 they satisfy the following Poisson brackets
\begin{equation*}
\lbrace S_1, S_2 \rbrace=-(S_2+S_3), \qquad \lbrace S_1, S_3 \rbrace=-(S_2+S_3),
\end{equation*}
i.e. we have the integrable Hamiltonian  system with symmetry Lie group
$\mathbf{\rom{3}} $  such that one can consider
its Hamiltonian
as:
\begin{equation*}
H=S_2=-s
\end{equation*}
and the invariants of the system are $(H, S_3).$
\end{eg}
\begin{eg}
  Lie group
$\mathbf{\rom{4}}$

Considering
  the  Lie group
$\mathbf{\rom{4}}$  related to the  Lie algebra
${\rom{4}}$ with non-zero commutators
$$[X_1, X_2]=-(X_2-X_3), [X_1, X_3]=-X_3,$$
 it admits the Jacobi
structure as follows   (see Table 1 ):
\begin{equation*}\label{}
\mathbf{\Lambda_1}=\partial _x\wedge \partial _y
+\left( y-z \right) \partial _y\wedge \partial _z,\quad
\mathbf{E_1}=
-\partial _y-\partial _z
\end{equation*}
where
  $( x, y, z )$ is
the local coordinate system on  the Lie group  $\mathbf{\rom{4}}$.
Applying the Poissonization  (\ref{aghajon})  of the Jacobi manifold  $(\mathbf{\rom{4}}, \mathbf{\Lambda_1}, \mathbf{E_1}),$
  it leads to the Poisson manifold
  $(\mathbf{\rom{4}} \otimes \mathbb{R}, P_1)$
    with
    the   local  coordinate system  $ (x, y, z, s) $ and
     the non-degenerate Poisson structure
\begin{equation*}
P_1:\qquad
\lbrace x,y   \rbrace={e}^{-s}\qquad
\lbrace  y,z \rbrace= ( y-z ){e}^{-s},\qquad
\lbrace  y,s \rbrace={e}^{-s},\qquad
\lbrace  z,s \rbrace={e}^{-s}.
\end{equation*}
Now one can find the following Darboux
coordinates:
\begin{equation*}
q_1=x,\qquad
q_2=y
,\qquad
p_1= \left( y-z \right) {e}^{s}
,\qquad
p_2={e}^{s}
\end{equation*}
such that they satisfy in the following canonical  Poisson brackets:
\begin{equation*}
\lbrace q_1,p_1   \rbrace=1,\qquad
\lbrace q_2,p_2  \rbrace=1.
\end{equation*}
We now consider that the Lie algebra $\rom{4}$ is realized by means of  smooth transformations  on the
phase space $\mathbb{  R}^4$  with the canonical  coordinate $(q_1, q_2, p_1, p_2 ) $
\begin{equation*}
 S_i=X_i(q_1, q_2, p_1,p_2 ),
\end{equation*}
where
 the  differential operator of
$ p_1=-\frac{\partial}{\partial q_1}$ and $ p_2=-\frac{\partial}{\partial q_2}$
(quantum mechanical realization)
are
 the conjugate
 momentums
  to
$ q_1$
 and
 $ q_2$, respectively.
Using the  results of
 \textcolor{red}{ [\ref{realiz}]} ( see Table 2)
 one can get the $S_i$
% (see relation (\ref{canon}))
 as follows:
\begin{equation*}
 S_1=q_1(q_2-1)p_1 +q_{2}^{2} p_2=x(y-1)  \left( y-z \right)  {e}^{s}+y^2e^{s},
\end{equation*}

\begin{equation*}
S_2=-p_1= \left( -y+z \right)  {e}^{s},\qquad
S_3=-q_2 p_1=y   \left( -y+z \right) \ {e}^{s}.
\end{equation*}
In this  way, now
applying relation (\ref{canon}),
one can show that
 they satisfy the following Poisson brackets
\begin{equation*}
\lbrace S_1, S_2 \rbrace=-(S_2-S_3), \qquad \lbrace S_1, S_3 \rbrace=-S_3,
\end{equation*}
i.e. we have the integrable Hamiltonian  system with symmetry Lie group
$\mathbf{\rom{4}} $  such that one can consider
its Hamiltonian
as:
\begin{equation*}
H=S_2=   \left( -y+z \right)  {e}^{s}
\end{equation*}
and the invariants of the system are $(H, S_3).$
\end{eg}
\begin{eg}
 Lie group
$\mathbf{\rom{6}_0}$

Considering
  the  Lie group
$\mathbf{\rom{6}_0}$  related to the  Lie algebra
${\rom{6}}_0$
with non-zero commutators
$$[X_1, X_3]=X_2, \quad [X_2, X_3]=X_1,$$
 it admits the Jacobi
structure as follows  (see Table 1):
\begin{equation*}
\mathbf{\Lambda_2}=-\sinh \left( z \right)\partial _x\wedge \partial _z+
\cosh \left( z \right) \partial _y\wedge \partial _z,\qquad
\mathbf{E_2}=
-\cosh \left( z \right) \partial _x+
 \sinh \left( z \right) \partial _y,
\end{equation*}
where
  $( x, y, z )$ is
the local coordinate system on  the Lie group  $\mathbf{\rom{6}_0}$.
Applying the Poissonization  (\ref{aghajon}) of the Jacobi manifold  $(\mathbf{\rom{6}_0}, \mathbf{\Lambda_2}, \mathbf{E_2}),$
  it leads to the Poisson manifold
  $(\mathbf{\rom{6}_0} \otimes \mathbb{R}, P_2)$
    with
    the   local  coordinate system  $ (x, y, z, s) $ and
     the non-degenerate Poisson structure
\begin{equation*}
P_2:\qquad
\lbrace x,z  \rbrace=-{e}^{-s} \sinh \left( z \right) ,\qquad
\lbrace  x,s \rbrace={e}^{-s}\cosh \left( z \right) ,\qquad
\lbrace  y,z \rbrace={e}^{-s} \cosh \left( z \right)  ,\qquad
\lbrace  y,s \rbrace=-{e}^{-s} \sinh \left( z \right).
\end{equation*}
Now one can find the following Darboux
coordinates:
\begin{equation*}
q_1=x,\qquad
q_2=y
,\qquad
p_1=\cosh \left( z \right) {e}^{s}
 ,\qquad
p_2=\sinh \left( z
 \right) {e}^{s}
\end{equation*}
such that they satisfy in the following canonical  Poisson brackets:
\begin{equation*}
\lbrace q_1,p_1   \rbrace=1,\qquad
\lbrace q_2,p_2  \rbrace=1.
\end{equation*}
We now consider that the Lie algebra $\rom{6}_0$ is realized by means of  smooth transformations  on the
phase space $\mathbb{  R}^4$  with the canonical  coordinate $(q_1, q_2, p_1, p_2 ) $
\begin{equation*}
 S_i=X_i(q_1, q_2, p_1,p_2 ),
\end{equation*}
where
 the  differential operator of
$ p_1=-\frac{\partial}{\partial q_1}$ and $ p_2=-\frac{\partial}{\partial q_2}$
(quantum mechanical realization)
are
 the conjugate
 momentums
  to
$ q_1$
 and
 $ q_2$, respectively.
Using the  results of
 \textcolor{red}{ [\ref{realiz}]} ( see Table 2)
 one can get the $S_i$
% (see relation (\ref{canon}))
 as follows:

\begin{equation*}
S_1=-p_1= -\cosh \left( z \right) {e}^{s}
,\qquad
S_2=-p_2=-\sinh \left( z
 \right) {e}^{s}
\end{equation*}

\begin{equation*}
S_3=- q_2 p_1-q_1p_2=
-y\cosh \left( z \right) {e}^{s}
-x\sinh \left( z
 \right) {e}^{s}.
\end{equation*}
In this  way, now
applying relation (\ref{canon}),
one can show that
 they satisfy the following Poisson brackets
\begin{equation*}
\lbrace S_1, S_3 \rbrace=S_2, \qquad \lbrace S_2, S_3 \rbrace=S_1,
\end{equation*}
i.e. we have the integrable Hamiltonian  system with symmetry Lie group
$\mathbf{\rom{6}_0} $  such that one can consider
its Hamiltonian
as:
\begin{equation*}
H=S_1=-\cosh \left( z \right) {e}^{s}
\end{equation*}
and the invariants of the system are $(H, S_2).$
\end{eg}
\begin{eg}
  Lie group
$\mathbf{\rom{7}_0}$

Considering
  the  Lie group
$\mathbf{\rom{7}_0}$ related to the  Lie algebra
${\rom{7}}_0$ with non-zero commutators
$$[X_1, X_3]=-X_2, [X_2, X_3]=X_1,$$
 it admits the Jacobi
structure as follows  (see Table 1 ):
\begin{align*}\label{}
\mathbf{\Lambda_1}=-\sin \left(z \right) \partial _x\wedge \partial _z+\cos \left(z \right)\partial _y\wedge \partial _z,\quad
\quad
\mathbf{ E_1}=
-\cos \left( z \right) \partial _x-\sin \left( z \right)
\partial _y,
\end{align*}
where
  $( x, y, z )$ is
the local coordinate system on  the Lie group  $\mathbf{\rom{7}_0}$.
Applying the Poissonization  (\ref{aghajon})  of the Jacobi manifold  $(\mathbf{\rom{7}_0}, \mathbf{\Lambda_1}, \mathbf{E_1}),$
  it leads to the Poisson manifold
  $(\mathbf{\rom{7}_0} \otimes \mathbb{R}, P_1)$
    with
    the   local  coordinate system  $ (x, y, z, s) $ and
     the non-degenerate Poisson structure
\begin{equation*}
P_1:\qquad
\lbrace x,z  \rbrace=-{e}^{-s} \sin \left( z \right)  ,\qquad
\lbrace  y,z \rbrace={e}^{-s} \cos \left( z \right)  ,\qquad
\lbrace x,s \rbrace={e}^{-s} \cos \left( z \right)  ,\qquad
 \lbrace  y,s \rbrace={e}^{-s} \sin \left( z \right).
\end{equation*}
Now one can find the following Darboux
coordinates:
\begin{equation*}
q_1=x,\qquad
q_2=y
,\qquad
p_1= \cos \left( z
 \right)  {e}^{s}
  ,\qquad
p_2={e}^{s} \sin
 \left( z \right)
\end{equation*}
such that they satisfy in the following canonical  Poisson brackets:
\begin{equation*}
\lbrace q_1,p_1   \rbrace=1,\qquad
\lbrace q_2,p_2  \rbrace=1.
\end{equation*}
We now consider that the Lie algebra $\rom{7}_0$ is realized by means of  smooth transformations  on the
phase space $\mathbb{  R}^4$  with the canonical  coordinate $(q_1, q_2, p_1, p_2 ) $
\begin{equation*}
 S_i=X_i(q_1, q_2, p_1,p_2 ),
\end{equation*}
where
 the  differential operator of
$ p_1=-\frac{\partial}{\partial q_1}$ and $ p_2=-\frac{\partial}{\partial q_2}$
(quantum mechanical realization)
are
 the conjugate
 momentums
  to
$ q_1$
 and
 $ q_2$, respectively.
Using the result results of
 \textcolor{red}{ [\ref{realiz}]} ( see Table 2)
 one can get the $S_i$
% (see relation (\ref{canon}))
 as follows:
\begin{equation*}
S_1=-p_1=-\cos \left( z
 \right)  {e}^{s}
,\qquad
S_2=-p_2=-{e}^{s}\sin
 \left( z \right)
\qquad
\end{equation*}
\begin{equation*}
S_3=- q_2 p_1+q_1p_2=
-y\cos \left( z
 \right) {e}^{s}
+x{e}^{s}\sin
 \left( z \right).
\end{equation*}
In this  way, now
applying relation (\ref{canon}),
one can show that
 they satisfy the following Poisson brackets
\begin{equation*}
\lbrace S_1, S_3 \rbrace=-S_2, \qquad \lbrace S_2, S_3 \rbrace=S_1,
\end{equation*}
i.e. we have the integrable Hamiltonian  system with symmetry Lie group
$\mathbf{\rom{7}_0} $  such that one can consider
its Hamiltonian
as:
\begin{equation*}
H=S_1=-\cos \left( z
 \right) {e}^{s}
\end{equation*}
and the invariants of the system are $(H, S_2).$
\end{eg}
\section{
  Bi-Hamiltonian systems
  by Jacobi structures on  real
  three-dimensional Lie groups
  }
 The study of bi-Hamiltonian
systems started with the pioneering work by Franco Magri \textcolor{red}{ [\ref{Magri}]}.
\begin{dfn}
 A pair $(P_1, P_2)$ of
Poisson structures on M
  is
said to be compatible
 if
 \textcolor{red}{ [\ref{Magri}]}
 \begin{equation}
 [P_1, P_1]=[P_2, P_2]=[P_1, P_2]=0,
 \end{equation}
 where $[.,.]$ is
 the Schouten–Nijenhuis bracket and the resulting bracket is
 the three vectors
 such that
their components
$[P_i, P_j]^{BCD}$
have the following forms:
 \footnote{
 Here
 we use  Einstein's summation
convention.
 }
 \begin{equation}
 [P_i, P_j]^{BCD}=P_i^{AB}\partial _A P_j^{CD}+
 P_i^{AD}\partial _A P_j^{BC}+
P_i^{AC}\partial _AP_j^{DB}.
 \end{equation}
\end{dfn}
\begin{dfn}
The manifold M equipped with  compatible Poisson structures $P_1$ and $P_2$ is called
the bi-Hamiltonian manifold.
\end{dfn}
\begin{dfn}
  A bi-Hamiltonian system is a dynamical system
possessing two compatible Hamiltonian formulations.
\end{dfn}

\begin{thm}\label{}
Let
 $(\Lambda^\prime,E^\prime)$
 and
$(\Lambda, E)$
 be
 two  Jacobi structures.
If there
exists an automorphism  $A  $  of  the Lie algebra
$\mathfrak{g}$
 such that
 \begin{equation}\label{p100}
\Lambda^\prime=A^t \Lambda A,
\end{equation}
 and
  \begin{equation}\label{r100}
E^{\prime^{e}}=E^{b} A_b^{\,\,e},
\end{equation}
then the Jacobi structures
$(\Lambda^\prime,E^\prime)$
 and
$(\Lambda, E)$
are equivalent.
\end{thm}\label{}
\begin{proof}
The proof is given in \textcolor{red}{ [\ref{hassan2}]}.
\end{proof}

\begin{lem}\label{}
Suppose that $(M\times \mathbb{R}, P_1)$
 is the Poissonization of the Jacobi manifold  $(M, \mathbf{\Lambda}, \mathbf{E})$
and
$(M\times \mathbb{R}, P_2)$
 is the Poissonization of the Jacobi manifold  $(M, \mathbf{\Lambda}^\prime, \mathbf{E}^\prime)$,
 and
 Jacobi structures $(\mathbf{\Lambda}, \mathbf{E})$ and $(\mathbf{\Lambda}^\prime, \mathbf{E}^\prime)$ are equivalent.
 In the general case,
   structures $P_1$ and $P_2$ are not
 compatible Poisson structures.
\end{lem}
\begin{proof}
By the definition of the Poissonization  of the Jacobi manifold  $(M, \mathbf{\Lambda}, \mathbf{E})$,
 we have
$$P_1=e^{-s}(\mathbf{\Lambda}+\partial_s\wedge\mathbf{ E})=
e^{-s}\mathbf{\Lambda}^{\rho\lambda} \partial_\rho\wedge \partial_\lambda
+e^{-s}\mathbf {E}^{\lambda} \partial_s
\wedge \partial_\lambda,
$$
such that
$P_1^{\rho\lambda}=e^{-s}\mathbf{\Lambda}^{\rho\lambda}
=e^{-s}
e_
{i}^{~\rho}
e_
{j}^{~\lambda}\Lambda^{ij}
$
and
$P_1^{s\lambda}=e^{-s}\mathbf {E}^{\lambda}=
e^{-s}e_
{k}^{~\lambda}E^k.$
Moreover,
by definition of the Poissonization  of the Jacobi manifold
 $(M, \mathbf{\Lambda}^\prime, \mathbf{E}^\prime)$, and using (\ref{p100}), (\ref{r100}),
 we have
\begin{equation*}
P_2=e^{-s}(\mathbf{\Lambda}^\prime+\partial_s\wedge\mathbf{ E}^\prime)=
e^{-s}(A^t\mathbf{\Lambda}A+\partial_s\wedge\mathbf{ E}A)=
e^{-s}
(A^t)_{~b}^{a}{\Lambda}^{bc}
A_
{c}^{~d}
e_
{a}^{~\mu}
e_
{d}^{~\nu}
\partial_\mu\wedge\partial_\nu
+e^{-s}
e_{f}^{~\mu}{ E}^k A_k^{~f}\partial_s\wedge\partial_\mu.
\end{equation*}
Here we employ also the notation
$P_2^{\mu \nu}=
e^{-s}
(A^t)_{~b}^{a}{\Lambda}^{bc}
A_
{c}^{~d}
e_
{a}^{~\mu}
e_
{d}^{~\nu}$
and
$P_2^{s\mu}=e^{-s}
e_{f}^{~\mu}{ E}^k A_k^{~f}.$
 One can
show that
$[P_1, P_2]^{BCD}\neq 0,$
 for $A=(s,\rho), B=\lambda, C=s, D=\mu$.
\end{proof}\\
 However, in the following,
we
will find examples  where
 structures $P_1$ and $P_2$
  are
 compatible.
Note that these
 Poisson structures
are obtained from
  the Poissonization of the
Jacobi structures.
Investigation  of general
Poisson structures
 on  real four-dimensional  Lie groups
 is previously studied in
\textcolor{red}{[\ref{Magriabd}]}.
\begin{eg}
  Lie group
$\mathbf{\rom{2}}$

In Example 3.1,
applying the Poissonization of the Jacobi manifold  $(\mathbf{\rom{2}}, \mathbf{\Lambda_1}, \mathbf{E_1}),$ we show that
  it leads to the Poisson manifold
  $(\mathbf{\rom{2}} \otimes \mathbb{R}, P_1)$
    with
    the   local  coordinate system  $ (x, y, z, s) $ and
     the non-degenerate Poisson structure:
\begin{align*}
P_1:\qquad
\lbrace x,z  \rbrace_1=-z{e}^{-s} ,\qquad
\lbrace  x,s \rbrace_1={e}^{-s},\qquad
\lbrace  y,z \rbrace_1={e}^{-s}.
\end{align*}
Now one can consider other representations of the same
 equivalence class of Jacobi structures on
 Lie group
 $\mathbf{\rom{2}}$ (see Table 1)
  for which after Poissonization one can obtain the  following non-degenerate Poisson structures
  $ (P_2, P_3, P_4)$ on
 $\mathbf{\rom{2}} \otimes \mathbb{R}$:
\begin{align*}
P_2:\qquad \lbrace x,z  \rbrace_2=(1-z){e}^{-s} ,\qquad
\lbrace  x,s \rbrace_2={e}^{-s},\qquad
\lbrace  y,z \rbrace_2={e}^{-s},\qquad\qquad\qquad\quad~\\
P_3:\qquad \lbrace x,y   \rbrace_3={e}^{-s},\qquad~~~
\lbrace x,z  \rbrace_3=-z{e}^{-s} ,\qquad
\lbrace  x,s \rbrace_3={e}^{-s},\qquad
\lbrace  y,z \rbrace_3={e}^{-s},\\
P_4:\qquad\lbrace x,y   \rbrace_4={e}^{-s},~~~~
\lbrace x,z  \rbrace_4=(1-z){e}^{-s} ,\qquad
\lbrace  x,s \rbrace_4={e}^{-s},\qquad
\lbrace  y,z \rbrace_4={e}^{-s},
\end{align*}
 such that the above structures   are    compatible with each other:
\begin{align*}
[P_i, P_j]=0
,\qquad i, j=1,2,3,4.
\end{align*}
Note that  there are other
 non-degenerate Poisson structures
on
 Lie group
 $\mathbf{\rom{2}} \otimes \mathbb{R}$.
 These structures
 can  be calculated
from relation
(\ref{Lie al2})
with
$E=0$.
After the simple calculations, one can obtain
 the following
 representation  (see Appendix and Table 3)
\begin{equation*}
\mathbf{ P_1^\prime}= \partial _{x}\wedge \partial_{ y}
+(1-z)\partial _{x}\wedge \partial_{ s}+
\partial _{y}\wedge \partial_{ s}
 +\partial _{z}\wedge \partial_{ s}
\end{equation*}
all other    compatible  Poisson structures on
$\mathbf{\rom{2}} \otimes \mathbb{R} ~ (i.e., \mathbf{P_2^\prime},  \cdots, \mathbf{P_{12}^\prime})$
 which are  not compatible with
$(P_1,P_2,P_3,P_4)$
are given in  Table 3 of the Appendix.
\end{eg}
\begin{eg}
  Lie group
$\mathbf{\rom{3}}$

In Example 3.2,
applying the Poissonization of the Jacobi manifold  $(\mathbf{\rom{3}}, \mathbf{\Lambda_1}, \mathbf{E_1}),$
 we show that
  it leads to the Poisson manifold
  $(\mathbf{\rom{3}} \otimes \mathbb{R}, P_1)$
    with
    the   local  coordinate system  $ (x, y, z, s) $ and
     the non-degenerate Poisson structure :
\begin{equation*}
P_1:\qquad
\lbrace x,z  \rbrace_1=e^{-s},\qquad
\lbrace  y,z \rbrace_1=\left( y+z \right){e}^{-s},\qquad
\lbrace  y,s \rbrace_1={e}^{-s},\quad
\lbrace  z,s \rbrace_1={e}^{-s}.
\end{equation*}
Now one can consider other representations of the
same
 equivalence class of Jacobi structures on
 Lie group
 $\mathbf{\rom{3}}$ (see Table 1)
  for which after Poissonization can obtain the  following non-degenerate Poisson structures
  $ (P_2, P_3, P_4)$ on
 $\mathbf{\rom{3}} \otimes \mathbb{R}$:
\begin{align*}
P_2:\qquad
\lbrace x,y   \rbrace_2={e}^{-s},\qquad
\lbrace y,z  \rbrace_2=-(y+z){e}^{-s} ,\qquad
\lbrace  y,s \rbrace_2={e}^{-s},\qquad
\lbrace  z,s \rbrace_2=-{e}^{-s}.\quad\\
P_3:\qquad
\lbrace x,z   \rbrace_3={e}^{-s},\qquad~~~
\lbrace y,z  \rbrace_3=(y+z+1){e}^{-s} ,\qquad
\lbrace  y,s \rbrace_3={e}^{-s},\qquad
\lbrace  z,s \rbrace_3={e}^{-s}.\\
P_4:\qquad
\lbrace x,y   \rbrace_4={e}^{-s},~~~~
\lbrace y,z  \rbrace_4=-(y+z-1){e}^{-s} ,\qquad
\lbrace  y,s \rbrace_4={e}^{-s},\qquad
\lbrace  z,s \rbrace_4=-{e}^{-s},
\end{align*}
such that the above structures   are    compatible with each other:
\begin{align*}
[P_i, P_j]=0
,\qquad i, j=1,2,3,4.
\end{align*}
\end{eg}
\begin{eg}
  Lie group
$\mathbf{\rom{4}}$

In Example 3.3,
applying the Poissonization of the Jacobi manifold  $(\mathbf{\rom{4}},\mathbf{ \Lambda_1}, \mathbf{E_1}),$ we show that it leads to the Poisson manifold
  $(\mathbf{\rom{4}} \otimes \mathbb{R}, P_1)$
    with
    the   local  coordinate system  $ (x, y, z, s) $ and
     the non-degenerate Poisson structure:
\begin{align*}
P_1:\qquad
\lbrace x,y   \rbrace_1={e}^{-s},\qquad
\lbrace  y,z \rbrace_1= ( y-z ){e}^{-s},\qquad
\lbrace  y,s \rbrace_1={e}^{-s},\qquad
\lbrace  z,s \rbrace_1={e}^{-s}.\qquad\qquad\qquad
\end{align*}
Now one can consider other representations of the
  same equivalence class of Jacobi structures on
 Lie group
 $\mathbf{\rom{4}}$ (see Table 1)
  for which after Poissonization  one can obtain the  following non-degenerate Poisson structures
  $ (P_2, P_3, P_4)$ on
 $\mathbf{\rom{4}} \otimes \mathbb{R}$:
\begin{align*}
P_2:\quad
\lbrace x,y   \rbrace_2={e}^{-s},\qquad
\lbrace  y,z \rbrace_2= {e}^{-s},\qquad
\lbrace  y,s \rbrace_2={e}^{-s},\qquad
\lbrace  z,s \rbrace_2={e}^{-s}.\qquad\qquad\qquad\qquad\qquad\qquad\qquad\\
P_3:\quad
\lbrace x,y   \rbrace_3={e}^{-s},\qquad
\lbrace  x,z \rbrace_3= {e}^{-s},\qquad
\lbrace  y,z \rbrace_3=( 2y-z ){e}^{-s},\qquad
\lbrace  y,s \rbrace_3={e}^{-s},\qquad
\lbrace  z,s \rbrace_3= 2{e}^{-s}.\qquad\\
P_4:\qquad
\lbrace x,y   \rbrace_4={e}^{-s},\quad
\lbrace  x,z \rbrace_4= {e}^{-s},\quad
\lbrace  y,z \rbrace_4=( 2y-z+1 ){e}^{-s},\quad
\lbrace  y,s \rbrace_4={e}^{-s},\qquad
\lbrace  z,s \rbrace_4= 2{e}^{-s},\qquad
\end{align*}
 such that the above structures   are    compatible with each other:
\begin{align*}
[P_i, P_j]=0
,\qquad i, j=1,2,3,4.
\end{align*}
\end{eg}
\begin{eg}
  Lie group
$\mathbf{\rom{6}_0}$

In Example 3.4,
applying the Poissonization of the Jacobi manifold  $(\mathbf{\rom{6}_0}, \mathbf{\Lambda_{2}}, \mathbf{E_{2}}),$ we show that
  it leads to the Poisson manifold
  $(\mathbf{\rom{6}_0} \otimes \mathbb{R}, P_2)$
    with
    the   local  coordinate system  $ (x, y, z, s) $ and
     the non-degenerate Poisson structure:
\begin{align*}
P_2:\qquad
\lbrace x,z  \rbrace_2=-{e}^{-s} sinh(z),\qquad
\lbrace y,z  \rbrace_2={e}^{-s} cosh(z),\qquad
\lbrace  x,s \rbrace_2={e}^{-s} cosh(z),\qquad
\lbrace  y,s \rbrace_2=-{e}^{-s}sinh(z).
\end{align*}
Now one can consider other representations of the  same
 equivalence class of Jacobi structures on
 Lie group
 $\mathbf{\rom{6}_0}$ (see Table 1)
  for which after Poissonization one can obtain the  following non-degenerate Poisson structures
  $ (P_1, P_3, P_4)$ on
 $\mathbf{\rom{6}_0} \otimes \mathbb{R}$:
\begin{align*}
P_1:\qquad
\lbrace x,z  \rbrace_1={e}^{-s}  cosh(z),\qquad
\lbrace y,z  \rbrace_1=-{e}^{-s}sinh(z),\qquad
\lbrace  x,s \rbrace_1=-{e}^{-s} sinh(z),\qquad
\lbrace  y,s \rbrace_1={e}^{-s} cosh(z),\qquad\qquad\qquad\quad\qquad\qquad\qquad\quad
\\
P_3:\qquad~
\lbrace x,y   \rbrace_3={e}^{-s},\qquad\qquad\quad
\lbrace x,z  \rbrace_3=e^{-s}cosh(z) ,\qquad
\lbrace  y,z \rbrace_3=-e^{-s}sinh(z),\qquad\qquad\qquad\qquad
\qquad\qquad\qquad\qquad\qquad\qquad\qquad\qquad\qquad
\\
\lbrace  x,s \rbrace_3=-e^{-s}sinh(z),\qquad\qquad\qquad
\lbrace  y,s \rbrace_3=e^{-s}cosh(z),\qquad\qquad\qquad\qquad\qquad\qquad\qquad
\qquad\qquad\qquad\qquad\qquad\qquad\qquad\qquad\qquad
\\
P_4:\qquad\quad
\lbrace x,y   \rbrace_4={e}^{-s},\qquad\qquad
\lbrace x,z  \rbrace_4=-e^{-s}sinh(z) ,\qquad\qquad
\lbrace  y,z \rbrace_4=e^{-s}cosh(z),\qquad
\qquad\qquad\qquad\qquad\qquad\qquad\qquad\qquad\qquad\qquad\qquad
\\
\lbrace  x,s \rbrace_4=e^{-s}cosh(z),\qquad
\lbrace  y,s \rbrace_4=-e^{-s}sinh(z),\qquad\qquad\qquad\qquad\qquad\qquad\qquad\qquad\qquad
\qquad\qquad\qquad\qquad\qquad\qquad\qquad
\qquad\qquad\qquad
\end{align*}
 such that the above structures   are    compatible with each other:
\begin{align*}
[P_i, P_j]=0
,\qquad i, j=1,2,3,4.
\end{align*}
\end{eg}
\begin{eg}
 Lie group
$\mathbf{\rom{7}_0}$

In Example 3.5,
applying the Poissonization of the Jacobi manifold  $(\mathbf{\rom{7}_0}, \mathbf{\Lambda_{1}}, \mathbf{E_{1}}),$ we show that
  it leads to the Poisson manifold
  $(\mathbf{\rom{7}_0} \otimes \mathbb{R}, P_1)$
    with
    the   local  coordinate system  $ (x, y, z, s) $ and
     the non-degenerate Poisson structure:
\begin{align*}
P_1:\qquad
\lbrace x,z  \rbrace_1=-{e}^{-s} sin(z),\qquad
\lbrace y,z  \rbrace_1={e}^{-s} cos(z),\qquad
\lbrace  x,s \rbrace_1={e}^{-s} cos(z),\qquad
\lbrace  y,s \rbrace_1={e}^{-s}sin(z).
\end{align*}
Now one can consider other representations of the  same
 equivalence class of Jacobi structures on
 Lie group
 $\mathbf{\rom{7}_0}$ (see Table 1)
  for which after Poissonization  one can obtain the  following non-degenerate Poisson structures
  $ (P_2, P_3, P_4,P_5, P_6)$ on
 $\mathbf{\rom{7}_0} \otimes \mathbb{R}$:
\begin{align*}
P_2:\quad
\lbrace x,z  \rbrace_2={e}^{s}  cos(z),\quad
\lbrace y,z  \rbrace_2={e}^{-s}sin(z),\quad
\lbrace  x,s \rbrace_2={e}^{-s} sin(z),\quad
\lbrace  y,s \rbrace_2=-{e}^{-s} cos(z).\qquad\qquad\qquad\quad\quad
\\
P_3:\quad
\lbrace x,z   \rbrace_3={e}^{-s}( cos(z)-sin(z)),\qquad
\lbrace y,z  \rbrace_3={e}^{-s}( cos(z)+sin(z)),\qquad
\lbrace  x,s \rbrace_3={e}^{-s}( cos(z)+sin(z)),\qquad\\
\lbrace  y,s \rbrace_3=-{e}^{-s}( cos(z)-sin(z)).\qquad\qquad\qquad\qquad\qquad\qquad\qquad\qquad
\qquad\qquad\qquad\qquad\qquad\qquad\qquad\qquad~\\
P_4:\quad
\lbrace x,y   \rbrace_4={e}^{-s},\quad
\lbrace x,z  \rbrace_4=-{e}^{-s} sin(z),\quad
\lbrace y,z  \rbrace_4={e}^{-s} cos(z),\quad
\lbrace  x,s \rbrace_4={e}^{-s} cos(z),\quad
\lbrace  y,s \rbrace_4={e}^{-s}sin(z).
\\
P_5:\quad
\lbrace x,y  \rbrace_5={e}^{s} ,\quad\,
\lbrace x,z  \rbrace_5={e}^{s}  cos(z),\quad
\lbrace y,z  \rbrace_5={e}^{-s}sin(z),\quad
\lbrace  x,s \rbrace_5={e}^{-s} sin(z),\quad
\lbrace  y,s \rbrace_5=-{e}^{-s} cos(z).\quad
\\
P_6:\quad
\lbrace x,y  \rbrace_6={e}^{s} ,\qquad\qquad
\lbrace x,z   \rbrace_6={e}^{-s}( cos(z)-sin(z)),\qquad\quad\quad
\lbrace y,z  \rbrace_6={e}^{-s}( cos(z)+sin(z)),\qquad\qquad\qquad
\\
\lbrace  x,s \rbrace_6={e}^{-s}( cos(z)+sin(z)),\qquad
\lbrace  y,s \rbrace_6=-{e}^{-s}( cos(z)-sin(z)),
\end{align*}
 such that the above structures   are    compatible with each other:
\begin{align*}
[P_i, P_j]=0
,\qquad i, j=1,2,3,4,5,6.
\end{align*}
\end{eg}
  \begin{sidewaystable}
   {\bf Table 1}:
 {\small
 Jacobi structures on real three-dimensional
  Lie algebras and
  Lie groups
  for which after  Poissonization  we have a non-degenerate Poisson brackets on
  \qquad $M=G\otimes \mathbb{R}$.}

    \begin{tabular}{|  l| l| l|   l|  }
      \hline
          {\footnotesize Jacobi structures  on  Lie algebra ${\rom{2}}$}
          &
{\footnotesize
Representation of one
 equivalence class
}
&
{\footnotesize
Representation of one
 equivalence class
}
&
{\footnotesize  compatible Poisson  structures  on $M=\mathbf{\rom{2}} \otimes \mathbb{R}$}
\\

&
 {\footnotesize
 on  Lie algebra ${\rom{2}}$
} &
 {\footnotesize
 on Lie group $\mathbf{\rom{2}}$ }
&
 \\
 \hline
 $\Lambda=
\lambda_{{12}}\partial _{x}\wedge \partial_{ y}+\lambda_{{13}}\partial _{x}\wedge \partial_{ z}
+ \lambda_{{23}}\partial _{y}\wedge \partial_{ z}$
 &
$\Lambda_1=\partial _{y}\wedge \partial_{ z}$
&
$\mathbf{\Lambda_1}=-z\partial _{x}\wedge \partial_{ z}
+\partial _{y}\wedge \partial_{ z}$
&
$P_1=-ze^{-s}\partial _{x}\wedge \partial_{ z}
+e^{-s}\partial _{x}\wedge \partial_{ s}$
 \\
$E=-\lambda_{{23}}\partial _{x}$
&
$E_1=-\partial _{x} $
&
$\mathbf{E_1}=-\partial _{x} $
&
$\qquad +e^{-s}\partial _{y}\wedge \partial_{ z}$
\\
&&&\\
{\footnotesize  Comment:}
   $
\lambda_{23}\neq 0
$
&
 $\Lambda_2=\partial _{x}\wedge \partial_{ z}+\partial _{y}\wedge \partial_{ z}
$
&
$\mathbf{\Lambda_2}=(1-z)\partial _{x}\wedge \partial_{ z}
+\partial _{y}\wedge \partial_{ z}$
&
$P_2=(1-z)e^{-s}\partial _{x}\wedge \partial_{ z}
+e^{-s}\partial _{x}\wedge \partial_{ s}$
 \\

 &
 $E_2=- \partial _{x} $
&
$\mathbf{E_2}=- \partial _{x} $
&
$\qquad +e^{-s}\partial _{y}\wedge \partial_{ z}$
\\
&&&\\

&
$\Lambda_3=\partial _{x}\wedge \partial_{ y}
+\partial _{y}\wedge \partial_{ z}
$
&
$\mathbf{\Lambda_3}=\partial _{x}\wedge \partial_{ y}
-z\partial _{x}\wedge \partial_{ z}
+\partial _{y}\wedge \partial_{ z}

$
&
$P_3=e^{-s}\partial _{x}\wedge \partial_{ y}-ze^{-s}\partial _{x}\wedge \partial_{ z}$
 \\

 &
 $E_3= -\partial _{x} $
&
$\mathbf{E_3}= -\partial _{x} $
&
$\qquad + e^{-s}\partial _{x}\wedge \partial_{ s}
+e^{-s}\partial _{y}\wedge \partial_{ z}$
\\
&&&\\

&
$\Lambda_4=\partial _{x}\wedge \partial_{ y}+\partial _{x}\wedge \partial_{ z}
+\partial _{y}\wedge \partial_{ z}
$
&
$\mathbf{\Lambda_4}=\partial _{x}\wedge \partial_{ y}
+(1-z)\partial _{x}\wedge \partial_{ z}
$
&
$P_4=e^{-s}\partial _{x}\wedge \partial_{ y}+(1-z)e^{-s}\partial _{x}\wedge \partial_{ z}$
 \\

 &
 $E_4= -\partial _{x} $
 &
 $\qquad \,+\partial _{y}\wedge \partial_{ z}$
 &
 $\qquad +e^{-s}\partial _{x}\wedge \partial_{ s}
+e^{-s}\partial _{y}\wedge \partial_{ z}
$
 \\
 &
&
$\mathbf{E_4}= -\partial _{x} $
&
\\   \hline
      {\footnotesize Jacobi structures  on  Lie algebra ${\rom{3}}$}
      &
{\footnotesize
Representation of one
 equivalence class
}
&
{\footnotesize
Representation of one
 equivalence class
}
&
{\footnotesize  compatible Poisson  structures  on $M=\mathbf{\rom{3}} \otimes \mathbb{R}$}
\\
&
 {\footnotesize
 on  Lie algebra ${\rom{3}}$
} &
 {\footnotesize
 on Lie group $\mathbf{\rom{3}}$ }
&
 \\
  \hline
${\Lambda}=
\lambda_{{12}}\partial _{x}\wedge \partial_{ y}+\lambda_{{13}}\partial _{x}\wedge \partial_{ z}+
\lambda_{{23}}\partial _{y}\wedge \partial_{ z}$
&
$\Lambda_1=\partial _{x}\wedge \partial_{ z}
$
&
$\mathbf{\Lambda_1}=\partial _{x}\wedge \partial_{ z}
+(y+z)\partial _{y}\wedge \partial_{ z}$
&
$P_1=e^{-s}\partial _{x}\wedge \partial_{ z}
+(y+z)e^{-s}\partial _{y}\wedge \partial_{ z}-$
 \\
$E=(\lambda_{{13}}-
\lambda_{{12}})\partial _{y}-(\lambda_{{13}}-\lambda_{{12}})\partial _{z}
$
&
$E_1=\partial _{y}-\partial _{z} $
&
$\mathbf{E_1}=\partial _{y}-\partial _{z} $
&
$\qquad e^{-s}\partial _{y}\wedge \partial_{ s}
+e^{-s}\partial _{z}\wedge \partial_{ s}$
\\
&&&\\
 {\footnotesize Comment:}
$  \begin {array}{c}
\lambda_{12}\neq \pm \lambda_{13}
\end {array}$
&
 $\Lambda_2=\partial _{x}\wedge \partial_{ y}
$
&
$\mathbf{\Lambda_2}=\partial _{x}\wedge \partial_{ y}
-(y+z)\partial _{y}\wedge \partial_{ z}$
&
$P_2=e^{-s}\partial _{x}\wedge \partial_{ y}
-(y+z)e^{-s}\partial _{y}\wedge \partial_{ z}+$
 \\

 &
 $E_2=- \partial _{y}+\partial _{z} $
&
$\mathbf{E_2}=- \partial _{y}+\partial _{z} $
&
$\qquad e^{-s}\partial _{y}\wedge \partial_{ s}
-e^{-s}\partial _{z}\wedge \partial_{ s}$
\\
&&&\\

&
$\Lambda_3=\partial _{x}\wedge \partial_{ z}
+\partial _{y}\wedge \partial_{ z}
$
&
$\mathbf{\Lambda_3}=\partial _{x}\wedge \partial_{ z}
+(y+z+1)\partial _{y}\wedge \partial_{ z}
$
&
$P_3=e^{-s}\partial _{x}\wedge \partial_{ z}
+(y+z+1)e^{-s}\partial _{y}\wedge \partial_{ z}-$
 \\

 &
 $E_3= \partial _{y}-\partial _{z} $
&
$\mathbf{E_3}= \partial _{y}-\partial _{z} $
&
$\qquad  e^{-s}\partial _{y}\wedge \partial_{ s}
+e^{-s}\partial _{z}\wedge \partial_{ s}$
\\
&&&\\

&
$\Lambda_4=\partial _{x}\wedge \partial_{ y}
+\partial _{y}\wedge \partial_{ z}
$
&
$\mathbf{\Lambda_4}=\partial _{x}\wedge \partial_{ y}
-(y+z-1)\partial _{y}\wedge \partial_{ z}
$
&
$P_4=e^{-s}\partial _{x}\wedge \partial_{ y}
-(y+z-1)e^{-s}\partial _{y}\wedge \partial_{ z}+$
 \\

 &
 $E_4= -\partial _{y}+\partial _{z} $
&
$\mathbf{E_4}= -\partial _{y}+\partial _{z} $
&
$\qquad e^{-s}\partial _{y}\wedge \partial_{ s}
-e^{-s}\partial _{z}\wedge \partial_{ s}$
\\
\hline
 \end{tabular}
 \end{sidewaystable}

 \begin{sidewaystable}
  {\bf Table 1}:
 {\small
  continue.
  }

 \begin{tabular}{|  l| l| l|  l|  }
 \hline
 {\footnotesize Jacobi structures  on  Lie algebra ${\rom{4}}$}
&
{\footnotesize
Representation of one
 equivalence class
 }
&{\footnotesize
Representation of one
 equivalence class
}
&
{\footnotesize  compatible Poisson  structures  on $M=\mathbf{\rom{4}} \otimes \mathbb{R}$}
\\

&
{\footnotesize
on   Lie algebra ${\rom{4}}$ }
&
{\footnotesize
on Lie group $\mathbf{\rom{4}}$ }
&
\\
   \hline
   ${\Lambda}=
\lambda_{{12}}\partial _{x}\wedge \partial_{ y}+\lambda_{{13}}\partial _{x}\wedge \partial_{ z}+
\lambda_{{23}}\partial _{y}\wedge \partial_{ z}$
&
$\Lambda_1=\partial _{x}\wedge \partial_{ y}
$
&
$\mathbf{\Lambda_1}=\partial _{x}\wedge \partial_{ y}
+(y-z)\partial _{y}\wedge \partial_{ z}$
&
$P_1=e^{-s}\partial _{x}\wedge \partial_{ y}
+(y-z)e^{-s}\partial _{y}\wedge \partial_{ z}$
 \\
$E=-\lambda_{{12}}\partial _{y}-(\lambda_{{12}}+\lambda_{{13}})\partial _{z}$
&
$E_1=-\partial _{y}-\partial _{z} $
&
$\mathbf{E_1}=-\partial _{y}-\partial _{z} $
&
$\qquad e^{-s}\partial _{y}\wedge \partial_{ s}
+e^{-s}\partial _{z}\wedge \partial_{ s}$
\\
&&&\\
{\footnotesize Comment:}
$  \begin {array}{c}
\lambda_{12}\neq 0
\end {array}$
&
 $\Lambda_2=\partial _{x}\wedge \partial_{ y}
 +\partial _{y}\wedge \partial_{ z}
$
&
$\mathbf{\Lambda_2}=\partial _{x}\wedge \partial_{ y}
+(y-z+1)\partial _{y}\wedge \partial_{ z}$
&
$P_2=e^{-s}\partial _{x}\wedge \partial_{ y}
+e^{-s}\partial _{y}\wedge \partial_{ z}$
 \\

 &
 $E_2=- \partial _{y}-\partial _{z} $
&
$\mathbf{E_2}=- \partial _{y}-\partial _{z} $
&
$\qquad +e^{-s}\partial _{y}\wedge \partial_{ s}
+e^{-s}\partial _{z}\wedge \partial_{ s}$
\\
 &&&
 \\
&
$\Lambda_3=\partial _{x}\wedge \partial_{ y}
+\partial _{x}\wedge \partial_{ z}
$
&
$\mathbf{\Lambda_3}=\partial _{x}\wedge \partial_{ y}
+\partial _{x}\wedge \partial_{ z}
$
&
$P_3=e^{-s}\partial _{x}\wedge \partial_{ y}
+e^{-s}\partial _{x}\wedge \partial_{ z}$
 \\
 &
 &
 $\qquad+(2y-z)\partial _{y}\wedge \partial_{ z}$
 &
 $\qquad +(2y-z) e^{-s}\partial _{y}\wedge \partial_{ z}
+e^{-s}\partial _{y}\wedge \partial_{ s}$
 \\
 &
 $E_3= -\partial _{y}-2\partial _{z} $
&
$\mathbf{E_3}= \partial _{y}-2\partial _{z} $
&
$\qquad+2e^{-s}\partial _{z}\wedge \partial_{ s}$
\\
&&&\\
&
$\Lambda_4=\partial _{x}\wedge \partial_{ y}
+\partial _{x}\wedge \partial_{ z}
+\partial _{y}\wedge \partial_{ z}
$
&
$\mathbf{\Lambda_4}=\partial _{x}\wedge \partial_{ y}
+\partial _{x}\wedge \partial_{ z}
$
&
$P_4=e^{-s}\partial _{x}\wedge \partial_{ y}
+e^{-s}\partial _{x}\wedge \partial_{ z}$
 \\
 &
 &
 $\qquad+(2y-z+1)\partial _{y}\wedge \partial_{ z}$
 &
  $\qquad+e^{-s}(2y-z+1)\partial _{y}\wedge \partial_{ z}$
 \\
 &
 $E_4= -\partial _{y}-2\partial _{z} $
&
$\mathbf{E_4}= -\partial _{y}+\partial _{z} $
&
$\qquad e^{-s}\partial _{y}\wedge \partial_{ s}
+2e^{-s}\partial _{z}\wedge \partial_{ s}$
\\
\hline
{\footnotesize Jacobi structures  on  Lie algebra ${\rom{6}_0}$
}
&
{\footnotesize  Representation of one
 equivalence class  }
 &{\footnotesize  Representation of one
 equivalence class    }
&
{\footnotesize  compatible Poisson  structures  on $M=\mathbf{\rom{6}_0} \otimes \mathbb{R}$}
\\
&
on  Lie algebra ${\rom{6}_0}$
&
on   Lie group $\mathbf{\rom{6}_0}$
&
\\
   \hline
   ${\Lambda}=
\lambda_{{12}}\partial _{x}\wedge \partial_{ y}+\lambda_{{13}}\partial _{x}\wedge \partial_{ z}+
\lambda_{{23}}\partial _{y}\wedge \partial_{ z}$
&
$\Lambda_1=\partial _{x}\wedge \partial_{ z}
$
&
$\mathbf{\Lambda_1}=cosh(z)\partial _{x}\wedge \partial_{ z}
$
&
$P_1=e^{-s}cosh(z)\partial _{x}\wedge \partial_{ z}
-e^{-s}sinh(z)\partial _{y}\wedge \partial_{ z}$
 \\
 $E=-\lambda_{{23}}\partial _{x}-\lambda_{13}\partial _{y}
$
 &
&
$\qquad-sinh(z)\partial _{y}\wedge \partial_{ z}$
&
$\qquad -e^{-s}sinh(z)\partial _{x}\wedge \partial_{ s}
+e^{-s}cosh(z)\partial _{y}\wedge \partial_{ s}$
\\
&
$E_1=-\partial _{y} $
&
$\mathbf{E_1}=sinh(z)\partial _{x}-cosh(z)\partial _{y} $
&

\\
{\footnotesize Comment:}
$  \begin {array}{c}
\lambda_{13}\neq \pm
\lambda_{23}
\end {array}$
&&&\\

 &
$\Lambda_2=\partial _{y}\wedge \partial_{ z}
$
&
$\mathbf{\Lambda_2}=-sinh(z)\partial _{x}\wedge \partial_{ z}
$
&
$P_2=-e^{-s}sinh(z)\partial _{x}\wedge \partial_{ z}
+e^{-s}cosh(z)\partial _{y}\wedge \partial_{ z}$
 \\
 &
&
$\qquad+cosh(z)\partial _{y}\wedge \partial_{ z}$
&
$\qquad +e^{-s}cosh(z)\partial _{x}\wedge \partial_{ s}
-e^{-s}sinh(z)\partial _{y}\wedge \partial_{ s}$
\\
&
$E_2=-\partial _{x} $
&
$\mathbf{E_2}=-cosh(z)\partial _{x}+sinh(z)\partial _{y} $
&

  \\
&
$\Lambda_3=\partial _{x}\wedge \partial_{ y}
+\partial _{x}\wedge \partial_{ z}
$
&
$\mathbf{\Lambda_3}=\partial _{x}\wedge \partial_{ y}+cosh(z)\partial _{x}\wedge \partial_{ z}
$
&
$P_3=e^{-s}\partial _{x}\wedge \partial_{ y}
+e^{-s}cosh(z)\partial _{x}\wedge \partial_{ z}
$
 \\
 &
&
$\qquad-sinh(z)\partial _{y}\wedge \partial_{ z}$
&
$\qquad-e^{-s}sinh(z)\partial _{y}\wedge \partial_{ z}$
\\
&&&
$\qquad -e^{-s}sinh(z)\partial _{x}\wedge \partial_{ s}
+e^{-s}cosh(z)\partial _{y}\wedge \partial_{ s}$
\\
&
$E_3=-\partial _{y} $
&
$\mathbf{E_3}=sinh(z)\partial _{x}-cosh(z)\partial _{y} $
&

\\
&&&\\
&
$\Lambda_4=\partial _{x}\wedge \partial_{ y}
+\partial _{y}\wedge \partial_{ z}
$
&
$\mathbf{\Lambda_4}=
\partial _{x}\wedge \partial_{ y}
-sinh(z)\partial _{x}\wedge \partial_{ z}
$
&
$P_4=e^{-s}\partial _{x}\wedge \partial_{ y}
-e^{-s}sinh(z)\partial _{x}\wedge \partial_{ z}
$
 \\
&
&
$\qquad+cosh(z)\partial _{y}\wedge \partial_{ z}$
&
$\qquad+e^{-s}cosh(z)\partial _{y}\wedge \partial_{ z}$
\\
&
$E_4=-\partial _{x} $
&
$\mathbf{E_4}=-cosh(z)\partial _{x}+sinh(z)\partial _{y} $
&
$\qquad +e^{-s}cosh(z)\partial _{x}\wedge \partial_{ s}
-e^{-s}sinh(z)\partial _{y}\wedge \partial_{ s}$
\\
\hline
     \end{tabular}
    \end{sidewaystable}

  \begin{sidewaystable}
   {\bf Table 1}:
 {\small
  continue.
  }

 \begin{tabular}{|  l| l| l|  l| }
  \hline
   {\footnotesize Jacobi structures  on  Lie algebra ${\rom{7}_0}$}
  &
  {\footnotesize  Representation of one
 equivalence class} &
  {\footnotesize Representation of one
 equivalence class }&
  {\footnotesize  compatible Poisson  structures  on $M=\mathbf{\rom{7}_0} \otimes \mathbb{R}$}
 \\
 &
{\footnotesize
on  Lie algebra ${\rom{7}_0}$}
&
{\footnotesize
 on Lie group $\mathbf{\rom{7}_0}$}
&
\\
 \hline

 ${\Lambda}=
\lambda_{{12}}\partial _{x}\wedge \partial_{ y}+\lambda_{{13}}\partial _{x}\wedge \partial_{ z}+
\lambda_{{23}}\partial _{y}\wedge \partial_{ z}$
 &
$\Lambda_1=\partial _{y}\wedge \partial_{ z}
$
&
$\mathbf{\Lambda_1}=
-sin(z)\partial _{x}\wedge \partial_{ z}
$
&
$P_1=-e^{-s}sin(z)\partial _{x}\wedge \partial_{ z}
+e^{-s}cos(z)\partial _{y}\wedge \partial_{ z}
$
 \\

 $E=-\lambda_{{23}}\partial _{x}+\lambda_{13}\partial _{y}
$
 &
$E_1=-\partial _{x} $
&
$\qquad +cos(z)\partial _{y}\wedge \partial_{ z}$
&
$\qquad +e^{-s}cos(z)\partial _{x}\wedge \partial_{ s}
+e^{-s}sin(z)\partial _{y}\wedge \partial_{ s}$
\\

&
&
$\mathbf{E_1}=-cos(z)\partial _{x}-sin(z)\partial _{y} $
&

\\
{\footnotesize Comment:}
$  \begin {array}{c}
\lambda_{13}^2+\lambda_{23}^2\neq 0
\end {array}$
 &
&

&
\\
&
$\Lambda_2=\partial _{x}\wedge \partial_{ z}
$
&
$\mathbf{\Lambda_2}=cos(z)\partial _{x}\wedge \partial_{ z}
$
&
$P_2=
+e^{-s}cos(z)\partial _{x}\wedge \partial_{ z}
+e^{-s}sin(z)\partial _{y}\wedge \partial_{ z}
$
 \\
 &
$E_2= \partial _{y} $
&
$\qquad +sin(z)\partial _{y}\wedge \partial_{ z}$
&
$\qquad +e^{-s}sin(z)\partial _{x}\wedge \partial_{ s}
-e^{-s}cos(z)\partial _{y}\wedge \partial_{ s}$
\\
&
&
$\mathbf{E_2}=-sin(z)\partial _{x}+cos(z)\partial _{y} $
&

\\
&&&\\
&
$\Lambda_3=\partial _{x}\wedge \partial_{ z}
+\partial _{y}\wedge \partial_{ z}
$
&
$\mathbf{\Lambda_3}=(cos(z)-sin(z))\partial _{x}\wedge \partial_{ z}
$
&
$P_3=
+e^{-s}(cos(z)-sin(z))\partial _{x}\wedge \partial_{ z}
$

 \\
 &
$E_3= -\partial _{x} +\partial _{y}$
&
$\qquad
+(sin(z)+cos(z))\partial _{y}\wedge \partial_{ z}
$&
$\qquad
+e^{-s}(sin(z)+cos(z))\partial _{y}\wedge \partial_{ z}
$
\\
&
&
$\mathbf{E_3}=(-sin(z)-cos(z))\partial _{x}$
&
$ \qquad+e^{-s}(sin(z)+cos(z))\partial _{x}\wedge \partial_{ s}$
\\
&
&
$\qquad+(cos(z)-sin(z))\partial _{y}$
&
$\qquad-e^{-s}(cos(z)-sin(z))\partial _{y}\wedge \partial_{ s}$
\\
&&&\\
&
$\Lambda_4=\partial _{x}\wedge \partial_{ y}
+\partial _{y}\wedge \partial_{ z}
$&
$\mathbf{\Lambda_4}=\partial _{x}\wedge \partial_{ y}
-sin(z)\partial _{x}\wedge \partial_{ z}
$
&
$P_4=
+e^{-s}\partial _{x}\wedge \partial_{ y}
-e^{-s}sin(z)\partial _{x}\wedge \partial_{ z}
$
 \\
 &
$E_4= -\partial _{x} $
&
$\qquad +cos(z)\partial _{y}\wedge \partial_{ z}$
&
$\qquad
+e^{-s}cos(z)\partial _{y}\wedge \partial_{ z}
$
\\
&
&
$\mathbf{E_4}=-cos(z)\partial _{x}-sin(z)\partial _{y} $
&$\qquad
+e^{-s}cos(z)\partial _{x}\wedge \partial_{ s}
+e^{-s}sin(z)\partial _{y}\wedge \partial_{ s}$
\\
&&&\\
&
$\Lambda_5=\partial _{x}\wedge \partial_{ y}
+\partial _{x}\wedge \partial_{ z}
$&
$\mathbf{\Lambda_5}=\partial _{x}\wedge \partial_{ y}
+cos(z)\partial _{x}\wedge \partial_{ z}
$
&
$P_5=+e^{-s}\partial _{x}\wedge \partial_{ y}
+e^{-s}cos(z)\partial _{x}\wedge \partial_{ z}
$
 \\
 &
$E_5= \partial _{y} $
&
$\qquad +sin(z)\partial _{y}\wedge \partial_{ z}$
&
$\qquad
+e^{-s}sin(z)\partial _{y}\wedge \partial_{ z}
+e^{-s}sin(z)\partial _{x}\wedge \partial_{ s}
$
\\
&
&
$\mathbf{E_5}=-sin(z)\partial _{x}+cos(z)\partial _{y} $

&
$\qquad-e^{-s}cos(z)\partial _{y}\wedge \partial_{ s}$
\\
&&&\\
&
$\Lambda_6=\partial _{x}\wedge \partial_{ y}
+\partial _{x}\wedge \partial_{ z}
$
&
$\mathbf{\Lambda_6}=\partial _{x}\wedge \partial_{ y}
$
&

 \\
 &
 $\qquad+\partial _{y}\wedge \partial_{ z}$
 &$\qquad+(cos(z)-sin(z))\partial _{x}\wedge \partial_{ z}$
&
$P_6=+e^{-s}\partial _{x}\wedge \partial_{ y}
$
\\&

&
$\qquad
+(sin(z)+cos(z))\partial _{y}\wedge \partial_{ z}
$
&
$
\qquad+e^{-s}(cos(z)-sin(z))\partial _{x}\wedge \partial_{ z}
$

\\
&
$E_6= -\partial _{x}+\partial _{y} $
&
$\mathbf{E_6}=(-sin(z)-cos(z))\partial _{x}$
&
$\qquad
+e^{-s}(sin(z)+cos(z))\partial _{y}\wedge \partial_{ z}
$
\\
&
&
$\qquad+(cos(z)-sin(z))\partial _{y}$
&$\qquad
+e^{-s}(sin(z)+cos(z))\partial _{y}\wedge \partial_{ z}
$
\\
&&&
$\qquad-e^{-s}(cos(z)-sin(z))\partial _{y}\wedge \partial_{ s}$
\\
\hline
      \end{tabular}
\end{sidewaystable}
  \begin{table}[t]
   {\bf Table 2}:
 {\small
  Realizations of some three-dimensional Lie algebras on
  $ \mathbb{R}^2$
  \textcolor{red}{ [\ref{realiz}]}}.

    \begin{tabular}{|l| l|  }
   \hline\hline
   Lie  algebra  with non-zero  commutation relations &
    Realization on
  $ \mathbb{R}^2$
  with coordinates $(q_1, q_2)$
  \\
   &\\
  \hline
&\\
$\rom{2}$  & $X_1=\partial _1,\,X_2= \partial _2,\, X_3=q_2\partial _1$\\
 $[X_2, X_3]=X_1$&\\
 &\\
 \hline
&\\
$ \rom{3} $  &$X_1= (q_1+q_2)\partial _1+(q_1+q_2)\partial _2,\, X_2=\partial _1,\,X_3=\partial _2  $\\
$[X_1, X_2]=-(X_2+X_3), [X_1, X_3]=-(X_2+X_3)$&\\
 &\\

 \hline
&\\
$ \rom{4}$  &$X_1=-q_1(q_2-1)\partial _1-q_2^2\partial _2,\,X_2=\partial _1,\,X_3= q_2\partial _1$\\
$[X_1, X_2]=-(X_2-X_3), [X_1, X_3]=-X_3$&\\
 &\\

  \hline
&\\
$ \rom{6}_0 $  &$X_1=\partial _1,\, X_2=\partial _2,\,X_3=q_2\partial _1+q_1\partial _2   $\\
$[X_1, X_3]=X_2, [X_2, X_3]=X_1$&\\
 &\\

\hline
&\\
 $\rom{7}_0$  &
 $X_1=\partial _1,\,X_2=\partial _2,\,X_3=q_2\partial _1-q_1\partial _ 2$\\
$[X_1, X_3]=-X_2, [X_2, X_3]=X_1$&\\
 &\\
\hline
        \end{tabular}
       % \end{center}
 \end{table}
 \newpage
  \section*{Appendix:
 Other  non-degenerate
compatible Poisson structures on
  the real four-dimensional   Lie group
$\mathbf{\rom{2}} \otimes \mathbb{R}$  }
%The first step is to show that
In this appendix,
we describe
 the details
for obtaining  Poisson structures
 on   real four-dimensional  Lie algebra
  ${\rom{2}} \oplus \mathbb{R}$.
We also
 obtain  the
  non-degenerate compatible
 Poisson
 structures on the related Lie group.
Note that in the classification of these Poisson structures, some of the structures are equivalent, and
 therefore define an equivalence relation and apply the following theorem:
\begin{thm}\label{poisson hassan}
Two Poisson structures
$P$
 and
$P^\prime$
 are equivalent if  there exists
 $A \in Aut(\mathfrak{g}), $
$ ( $ i.e.,
 automorphism group of the Lie algebra
$\mathfrak{g})$
 such that
 \begin{equation}\label{kaka}
P^\prime=\mathcal{A}^t P \mathcal{A},
\end{equation}
\end{thm}\label{}
\begin{proof}
The proof is given in \textcolor{red}{ [\ref{hassan2}]}.
\end{proof}

\subsection*{ Lie algebra ${\rom{2}} \oplus \mathbb{R}\cong  A_{3,1}\oplus A_1$}
We
first assume  the matrix form of the
Poisson structure
on   real four-dimensional  Lie algebra
 ${\rom{2}} \oplus \mathbb{R}$ as follows:
\begin{equation}\label{BARGH}
P=
\left( \begin {array}{cccc} 0&p_{{12}}&p_{{13}}&p_{{14}}
\\ \noalign{\medskip}-p_{12}&0&p_{23}&p_{24}
\\ \noalign{\medskip}-p_{{13}}&-p_{23}&0&p_{34}
\\ \noalign{\medskip}-p_{{14}}&-p_{24}&-p_{34}&0
\end {array}
 \right),
\end{equation}
where
$p _{ij} $
are arbitrary  real constants.
Using
$(\ref{rep}),$
we obtain
 adjoint representations
$\chi_i$
of the Lie algebra  ${\rom{2}} \oplus \mathbb{R}$:
\begin{equation}
\label{miroz}
\chi_1=\chi_4= 0,~
\chi_2=\left( \begin {array}{cccc} 0&0&0&0\\
\noalign{\medskip}0&0&0&0
\\ \noalign{\medskip}-1&0&0&0
\\ \noalign{\medskip}0&0&0&0
\end {array} \right),~
\chi_3=\left( \begin {array}{cccc} 0&0&0&0\\
\noalign{\medskip}1&0&0&0
\\ \noalign{\medskip}0&0&0&0
\\ \noalign{\medskip}0&0&0&0
\end {array} \right),~
\end{equation}
and antisymmetric matrices
${\cal Y}_i$
 :
\begin{equation}\label{antisym}
{\cal Y}_1= \left(\begin {array}{cccc}
0&0&0&0
\\ \noalign{\medskip}0&0&-1&0
\\ \noalign{\medskip}0&1&0&0
\\ \noalign{\medskip}0&0&0&0
\end {array} \right),~
{\cal Y}_2=
{\cal Y}_3=
{\cal Y}_4=0.
\end{equation}
Substituting  $E=0$ in (13), we have:
\begin{equation}\label{POISSON HASSAN}
 {P^{ce}}({\chi ^t}_cP ) + P {\cal Y}^eP  + (P \chi _b){P ^{be}}    = 0.
\end{equation}
Now inserting
 ($\ref{BARGH}$)-$( \ref{antisym})$   in ($\ref{POISSON HASSAN}$), one can obtain the  Poisson structures
for  the Lie algebra
${\rom{2}} \oplus \mathbb{R}$
as follows:
\begin{equation}\label{olampic}
P=
\left( \begin {array}{cccc}
0&p_{{12}}&p_{{13}}&p_{{14}}
\\ \noalign{\medskip}-p_{12}&0&0&p_{24}
\\ \noalign{\medskip}-p_{{13}}&0&0&p_{34}
\\ \noalign{\medskip}-p_{14}&-p_{24}&-p_{34}&0
\end {array}
 \right).
\end{equation}
%Another representation of the  Poisson structuresfor the Lie algebra $\mathbf{\rom{3}} %\times\mathbb{R}$:
% the partition consists of the equivalence classe
Then
applying the automorphism group of the Lie algebra ${\rom{2}} \oplus \mathbb{R}$
 \textcolor{red}{ [\ref{SEFIDREZAEI}]}
 \begin{equation}\label{IIautomorphism}
\mathcal
{A}
=
 \left( \begin {array}{cccc} a_{{22}}a_{{33}}-a_{{23}}a_{{32}}&0&0&0
\\ \noalign{\medskip}a_{{21}}&a_{{22}}&a_{{23}}&a_{{24}}
\\ \noalign{\medskip}a_{{31}}&a_{{32}}&a_{{33}}&a_{{34}}
\\ \noalign{\medskip}a_{{41}}&0&0&a_{{44}}\end {array} \right),
\end{equation}
where
$a _{ij}\in \mathbb{R}.$
%Inserting the above automorphism
% and
%
 Using
 $( \ref{kaka})$
with two equivalent Poisson structure
$P$ and $P^\prime$ from
$( \ref{olampic})$,
we get
 det$\mathcal{A} =
\dfrac{p^\prime_{{12}}p^\prime_{{34}}-p^\prime_{{13}}p^\prime_{{24}}}
{p_{{12}}p_{{34}}-p_{{13}}p_{{24}}}$.
Since we
must have
det$\mathcal{A}
\neq 0,
$
 it follows that
$p^\prime_{{12}}p^\prime_{{34}}-p^\prime_{{13}}p^\prime_{{24}}\neq 0.$
 % have to be non-zero
  Moreover,
  det$\mathcal{A}$
  does not depend on parameters
  $ p^\prime_{14}$;
  thus
 these parameters can take any
value.

By using  Theorem $\ref{poisson hassan}$, we show that the Poisson structure $P$ consists of the following
structures in one
 equivalence class:\\

(1)
 If
 $p^\prime_{12}=1,
  p^\prime_{34}=1,
  \, p^\prime_{13}=0,\,
 p^\prime_{24}=1,\,
 p^\prime_{14}=1
 $,
 then
   \begin{equation*}\label{}
 \mathcal
{A}
=   \left(\begin {array}{cccc}
    -{\frac {p_{{24}}}{a_{{32}} \left( p_{{12
}}p_{{34}}-p_{{13}}p_{{24}} \right) }}&0&0&0\\
 \noalign{\medskip}a_{{
21}}&{\frac {-{a_{{32}}}^{3}p_{{12}}{p_{{34}}}^{2}+{a_{{32}}}^{3}p_{{
13}}p_{{24}}p_{{34}}+{p_{{24}}}^{2}}{p_{{24}}{a_{{32}}}^{2} \left( p_{
{12}}p_{{34}}-p_{{13}}p_{{24}} \right) }}&{\frac {p_{{24}}}{{a_{{32}}}
^{2} \left( p_{{12}}p_{{34}}-p_{{13}}p_{{24}} \right) }}&a_{{24}}
\\
\noalign{\medskip}a_{{31}}&a_{{32}}&0&a_{34}\\
 \noalign{\medskip}-{\frac {p_{
{12}}}{a_{{32}} \left( p_{{12}}p_{{34}}-p_{{13}}p_{{24}} \right) }}&0&0
&{\frac {{a_{{32}}}^{2} \left( p_{{12}}p_{{34}}-p_{{13}}p_{{24}}
 \right) }{{p_{{24}}}^{2}}}\end {array} \right),
 \end{equation*}
   where
   \begin{equation*}
   a_{34}=-{\frac { \left(  \left( p_{{
12}}p_{{34}}-p_{{13}}p_{{24}} \right)  \left( a_{{21}}p_{{24}}+a_{{31}
}p_{{34}} \right) {a_{{32}}}^{2}-a_{{32}}p_{{14}}p_{{24}}-{p_{{24}}}^{
2} \right) a_{{32}}}{{p_{{24}}}^{2}}}
\end{equation*}
with
det$\mathcal{A} =
\dfrac{1}
{p_{{12}}p_{{34}}-p_{{13}}p_{{24}}}$
 and
the Poisson structure $P$ is  equivalent to
\begin{equation}\label{coordinateii}
P_1^\prime=\partial _x\wedge \partial _y+\partial _z\wedge \partial _s+\partial _y\wedge \partial _s+\partial _x\wedge \partial _s,
\end{equation}

(2)
 If
 $p^\prime_{12}=1,
  p^\prime_{34}=1,
  \, p^\prime_{13}=1,\,
 p^\prime_{24}=0,\,
 p^\prime_{14}=0
 $,
 then
 the Poisson structure $P$ is  equivalent to
\begin{equation*}\label{}
P_2^\prime=\partial _x\wedge \partial _y+\partial _z\wedge \partial _s+\partial _x\wedge \partial _z,
\end{equation*}

(3)
 If
 $p^\prime_{12}=1,
  p^\prime_{34}=1,
  \, p^\prime_{13}=0,\,
 p^\prime_{24}=1,\,
 p^\prime_{14}=0
 $,
 then
the Poisson structure $P$ is  equivalent to
\begin{equation*}\label{}
P_3^\prime=\partial _x\wedge \partial _y+\partial _z\wedge \partial _s+\partial _y\wedge \partial _s,
\end{equation*}

(4)
 If
 $p^\prime_{12}=1,
  p^\prime_{34}=1,
  \, p^\prime_{13}=0,\,
 p^\prime_{24}=0,\,
 p^\prime_{14}=1
 $,
 then
 the Poisson structure $P$  is  equivalent to
\begin{equation*}\label{}
P_4^\prime=\partial _x\wedge \partial _y+\partial _z\wedge \partial _s+\partial _x\wedge \partial _s,
\end{equation*}

(5)
 If
 $p^\prime_{12}=1,
  p^\prime_{34}=1,
  \, p^\prime_{13}=1,\,
 p^\prime_{24}=0,\,
 p^\prime_{14}=1
 $,
 then
 the Poisson structure $P$  is  equivalent to
\begin{equation*}\label{}
P_5^\prime=\partial _x\wedge \partial _y+\partial _z\wedge \partial _s+\partial _x\wedge \partial _z+
\partial _x\wedge \partial _s,
\end{equation*}

(6)
If
 $p^\prime_{12}=1,
  p^\prime_{34}=1,
  \, p^\prime_{13}=0,\,
 p^\prime_{24}=0,\,
 p^\prime_{14}=0
 $,
 then
 the Poisson structure $P$ is  equivalent to
\begin{equation*}\label{ashora}
P_6^\prime=\partial _x\wedge \partial _y+\partial _z\wedge \partial _s,
\end{equation*}

  (7)
 If
 $p^\prime_{12}=0,
  p^\prime_{34}=0,
  \, p^\prime_{13}=1,\,
 p^\prime_{24}=1,\,
 p^\prime_{14}=0
 $,
 then
 the Poisson structure $P$ is  equivalent to
\begin{equation*}\label{ashora}
P_7^\prime=\partial _x\wedge \partial _z+\partial _y\wedge \partial _s,
\end{equation*}

(8)
 If
 $p^\prime_{12}=1,
  p^\prime_{34}=0,
  \, p^\prime_{13}=1,\,
 p^\prime_{24}=1,\,
 p^\prime_{14}=0
 $,
 then
 the Poisson structure $P$ is  equivalent to
\begin{equation*}\label{}
P_8^\prime=\partial _x\wedge \partial _y+\partial _x\wedge \partial _z+\partial _y\wedge \partial _s,
\end{equation*}

(9)
 If
 $p^\prime_{12}=0,
  p^\prime_{34}=1,
  \, p^\prime_{13}=1,\,
 p^\prime_{24}=1,\,
 p^\prime_{14}=0
 $,
 then
the Poisson structure $P$ is  equivalent to
\begin{equation*}\label{}
P_9^\prime=\partial _z\wedge \partial _s+\partial _x\wedge \partial _z+\partial _y\wedge \partial _s,
\end{equation*}

  (10)
 If
 $p^\prime_{12}=0,
  p^\prime_{34}=0,
  \, p^\prime_{13}=1,\,
 p^\prime_{24}=1,\,
 p^\prime_{14}=1
 $,
 then
 the Poisson structure $P$ is  equivalent to
\begin{equation*}\label{ashora}
P_{10}^\prime=\partial _x\wedge \partial _z+\partial _y\wedge \partial _s+\partial _x\wedge \partial _s,
\end{equation*}

(11)
 If
 $p^\prime_{12}=1,
  p^\prime_{34}=0,
  \, p^\prime_{13}=1,\,
 p^\prime_{24}=1,\,
 p^\prime_{14}=1
 $,
 then
 the Poisson structure $P$ is  equivalent to
\begin{equation*}\label{}
P_{11}^\prime=\partial _x\wedge \partial _y+\partial _x\wedge \partial _z+\partial _y\wedge \partial _s+
\partial _x\wedge \partial _s,
\end{equation*}

(12)
 If
 $p^\prime_{12}=0,
  p^\prime_{34}=1,
  \, p^\prime_{13}=1,\,
 p^\prime_{24}=1,\,
 p^\prime_{14}=1
 $,
 then
the Poisson structure $P$ is  equivalent to
\begin{equation*}\label{}
P_{12}^\prime=\partial _z\wedge \partial _s+\partial _x\wedge \partial _z+\partial _y\wedge \partial _s+\partial _x\wedge \partial _s,
\end{equation*}

Now  these  Poisson structures     can  be converted to   the    Poisson   structures
 on the related Lie group.
To compute
these  Poisson  structures, we need to determine  the vielbein $e_{a}^{\;\;\mu}$ for Lie groups, and
in order to find  the vielbein $e_{a}^{\;\;\mu}$ for the Lie group
$\mathbf{\rom{2}} \otimes \mathbb{R}$,
 it is required to
calculate
 the  left-invariant one-forms on   the Lie group $\mathbf{\rom{2}} \otimes \mathbb{R}$ as follows
  \textcolor{red}{ [\ref{Mojaveri}]}:

 \begin{equation*}
 g^{-1}dg= e_{\;\;\mu}^{a}X_{a}dx^{\mu}=dx X_1+dy (X_2+z X_1)+dz X_3+ds X_4.
 \end{equation*}
 In view of the above relation, we get
\begin{equation*}\label{}
 e_{\;\;\mu}^{a}=\left( \begin {array}{cccc}
  1&z&0&0\\
 \noalign{\medskip}
 0&1&0&0
\\ \noalign{\medskip}0&0&1&0
\\ \noalign{\medskip}0&0&0&1
\end {array} \right).
\end{equation*}
Then one can obtain  the inverse of the vielbein $e_{\;\;\mu}^{a}$, that is $e_{a}^{\;\;\mu}$, for the Lie group $\mathbf{\rom{2}} \otimes \mathbb{R}$
\begin{equation}\label{e_a}
 e_{a}^{\;\;\mu}=\left( \begin {array}{cccc}
  1&-z&0&0\\
 \noalign{\medskip}
 0&1&0&0
\\ \noalign{\medskip}0&0&1&0
\\ \noalign{\medskip}0&0&0&1
\end {array} \right).
\end{equation}
Hence, substituting $(\ref{e_a})$ and $(\ref{coordinateii})$  in $(\ref{basis})$, one can calculate the   Poisson structure
$\mathbf{P_1^\prime}$   on the Lie
group
$\mathbf{\rom{2}}\otimes \mathbb{R}$:
\begin{equation*}
\mathbf{P_1^\prime}= \left( \begin {array}{cccc} 0&1&0&1-z\\ \noalign{\medskip}-1&0&0&1
\\ \noalign{\medskip}0&0&0&1\\ \noalign{\medskip}-1+z&-1&-1&0
\end {array} \right)
\end{equation*}
In the same way, one
can  obtain the
other
structures
$\mathbf{P_2^\prime},\cdots, \mathbf{P_{12}^\prime}$
on  the Lie group ${\mathbf{\rom{2}}\otimes \mathbb{R}}$
(see Table 3).
Note that  we  consider non-degenerate and compatible Poisson structures  on four-dimensional  Lie groups
$\mathbf{\rom{2}} \otimes \mathbb{R}$
(see Table 3).
 In the same way,
 we have obtained the
  Poisson structures  on   Lie groups
$\mathbf{\rom{3}} \otimes \mathbb{R},~\mathbf{\rom{4}} \otimes \mathbb{R}$ and $\mathbf{\rom{6}_0} \otimes \mathbb{R}$
but unfortunately all of them
 are degenerate.
Also,
for the
 Lie group
 $\mathbf{\rom{7}_0 }\otimes \mathbb{R}$
we have not found any  solution.
  \begin{table}[h]
 {\bf Table 3}:
 {\small
Other  non-degenerate
compatible
Poisson structures  on the four-dimensional  Lie algebra
${\rom{2}} \oplus \mathbb{R}$
 and its  Lie group
 $\mathbf{\rom{2}} \otimes \mathbb{R}$.
 }

    \begin{tabular}{|l| l|  l| }
    \hline
{\footnotesize Poisson structures  on  Lie algebra  ${\rom{2}}\oplus \mathbb{R}$}
 &
{\footnotesize
Representation of one
 equivalence class
 } &
 {\footnotesize
Representation of one
 equivalence class}
\\
&
{\footnotesize
  on  Lie algebra ${\rom{2}\oplus \mathbb{R}}$}
&
{\footnotesize
  on  Lie groups ${\mathbf{\rom{2}}\otimes \mathbb{R}}$}
\\
 \hline
 $P =p_{{12}}\partial _{x}\wedge \partial_{ y}+
 p_{{13}}\partial _{x}\wedge \partial_{ z}+
 $
&
${\footnotesize P_1^\prime=\partial _x\wedge \partial _y+\partial _z\wedge \partial _s+\partial _y\wedge \partial _s+\partial _x\wedge \partial _s}
$
&
$\mathbf{ P_1^\prime}= \partial _{x}\wedge \partial_{ y}
+(1-z)\partial _{x}\wedge \partial_{ s}+
$
\\
$ \qquad
p_{{14}}\partial _{x}\wedge \partial_{ s}
 +p_{{24}}\partial _{y}\wedge \partial_{ s}+
 $
&
&
$\qquad\quad \partial _{y}\wedge \partial_{ s}
+\partial _{z}\wedge \partial_{ s}
$
 \\
 $\qquad p_{{34}}\partial _{z}\wedge \partial_{ s}$&&\\
 &
 $
P_2^\prime=\partial _x\wedge \partial _y+\partial _z\wedge \partial _s+\partial _x\wedge \partial _z
$
 &
 $\mathbf{ P_2^\prime}=\partial _x\wedge \partial _y+\partial _z\wedge \partial _s+\partial _x\wedge \partial _z
$
 \\
 {\footnotesize Comment:}
 $p_{{12}}p_{{34}}-p_{{13}}p_{{24}}\neq 0$
 &
 &
 \\
 &
  $P_3^\prime=\partial _x\wedge \partial _y+\partial _z\wedge \partial _s+\partial _y\wedge \partial _s
$
 &
 $\mathbf{ P_3^\prime}= \partial _x\wedge \partial _y+\partial _z\wedge \partial _s+\partial _y\wedge \partial _s-
$
 \\
 &&
 $\qquad \quad z\partial _x\wedge \partial _s$
 \\
&
 &
 \\
 &
 $P_4^\prime=\partial _x\wedge \partial _y+\partial _z\wedge \partial _s+\partial _x\wedge \partial _s$
 &
 $\mathbf{P_4^\prime}=\partial _x\wedge \partial _y+\partial _z\wedge \partial _s+\partial _x\wedge \partial _s$
 \\
 & &
 \\
 &
  $P_5^\prime=\partial _x\wedge \partial _y+\partial _z\wedge \partial _s+\partial _x\wedge \partial _z+
\partial _x\wedge \partial _s$
 &
  $\mathbf{P_5^\prime}=\partial _x\wedge \partial _y+\partial _z\wedge \partial _s+\partial _x\wedge \partial _z+
$
 \\
 & &
 \qquad\quad $\partial _x\wedge \partial _s
 $
 \\
 &$P_6^\prime=\partial _x\wedge \partial _y+\partial _z\wedge \partial _s
$&
 $\mathbf{P_6^\prime}=\partial _x\wedge \partial _y+\partial _z\wedge \partial _s$
 \\
 &&\\
 &
 $P_7^\prime=\partial _x\wedge \partial _z+\partial _y\wedge \partial _s$
 &
 $\mathbf{P_7^\prime}=\partial _x\wedge \partial _z+\partial _y\wedge \partial _s-z\partial _x\wedge \partial _s$
 \\
 &&\\
 &
 $P_8^\prime=\partial _x\wedge \partial _y+\partial _x\wedge \partial _z+\partial _y\wedge \partial _s
$
 &
  $\mathbf{P_8^\prime}=\partial _x\wedge \partial _y+\partial _x\wedge \partial _z+\partial _y\wedge \partial _s
-$
 \\
 &&
 $\qquad \quad z\partial _x\wedge \partial _s$
 \\
 &&\\
 &
 $P_9^\prime=\partial _z\wedge \partial _s+\partial _x\wedge \partial _z+\partial _y\wedge \partial _s
$
 &
  $\mathbf{P_9^\prime}=\partial _z\wedge \partial _s+\partial _x\wedge \partial _z+\partial _y\wedge \partial _s-
$
 \\
 &&
 $\qquad \quad z\partial _x\wedge \partial _s$
 \\
 &&\\
 &
 $P_{10}^\prime=\partial _x\wedge \partial _z+\partial _y\wedge \partial _s+\partial _x\wedge \partial _s$
 &
  $\mathbf{P_{10}^\prime}=\partial _x\wedge \partial _z+\partial _y\wedge \partial _s+
$
 \\
 &&
 $\qquad\quad (1-z)\partial _x\wedge \partial _s$
 \\
 &&\\
 &
 $P_{11}^\prime=\partial _x\wedge \partial _y+\partial _x\wedge \partial _z+\partial _y\wedge \partial _s+
\partial _x\wedge \partial _s$
 &
 $\mathbf{P_{11}^\prime}=\partial _x\wedge \partial _y+\partial _x\wedge \partial _z+$
 \\
 &&
 \qquad\quad$\partial _y\wedge \partial _s+
(1-z)\partial _x\wedge \partial _s$
 \\
 &&\\
 &
 $P_{12}^\prime=\partial _z\wedge \partial _s+\partial _x\wedge \partial _z+\partial _y\wedge \partial _s+\partial _x\wedge \partial _s$
 &
  $\mathbf{P_{12}^\prime}=\partial _z\wedge \partial _s+\partial _x\wedge \partial _z+$
 \\
 &&
 \qquad\quad$\partial _y\wedge \partial _s+(1-z)\partial _x\wedge \partial _s$
 \\
 \hline
     \end{tabular}
 \end{table}
 
\end{document}